\documentclass[11pt]{article}
\pdfoutput=1 %

\usepackage[T1]{fontenc}
\usepackage{lmodern}

\usepackage[margin=1in]{geometry}
\usepackage{amsfonts,amssymb,amsmath,amsthm,mathtools,thmtools}
\usepackage{microtype,xspace,bm,floatrow}
\usepackage{enumitem}
\usepackage{braket}

\usepackage{doi}
\usepackage{hyperref}
\usepackage[dvipsnames]{xcolor}
\hypersetup{
	colorlinks=true,
	urlcolor=MidnightBlue,
	linkcolor=MidnightBlue,
	citecolor=MidnightBlue,
	unicode
}

\usepackage{bold-extra}

\usepackage{enumitem}
\setlist{  
  listparindent=\parindent,
  parsep=0pt,
}

\usepackage[capitalise,nameinlink,noabbrev]{cleveref}
\crefname{equation}{}{}

\usepackage{tikz}
\usepackage{graphicx}
\usetikzlibrary{positioning,arrows.meta,math,shapes.geometric,decorations.pathmorphing,decorations.pathreplacing,patterns,patterns.meta}

\usepackage[font=small,labelfont=bf]{caption} %
\usepackage{subcaption}

\usepackage[english]{babel}
\addto\extrasenglish{}
\addto\extrasenglish{}

\usepackage{lpic}

\makeatletter
\DeclareRobustCommand\bfseries{%
  \not@math@alphabet\bfseries\mathbf
  \fontseries\bfdefault\selectfont\boldmath}
\makeatother

\theoremstyle{definition}
\declaretheorem[name=Theorem]{theorem}
\declaretheorem[name=Main Result]{mainresult}
\declaretheorem[name=Proposition]{proposition}
\declaretheorem[name=Lemma,sibling=theorem]{lemma}
\declaretheorem[name=Corollary, sibling=theorem]{corollary}

\declaretheorem[name=Fact,sibling=theorem]{fact}

\newcommand{\poly}{\textrm{poly}}

\newcommand{\supp}{\operatorname{supp}}
\newcommand{\norm}[1]{\left\lVert #1 \right\rVert}

\newcommand{\car}{\mathsf{C}}
\newcommand{\sgp}{\sigma^{+}}
\newcommand{\sgm}{\sigma^{-}}
\newcommand{\cre}[2]{\sigma^{+}_{#1,\,#2}}
\newcommand{\ann}[2]{\sigma^{-}_{#1,\,#2}}
\newcommand{\proj}[1]{\lvert #1\rangle\!\langle #1\rvert}

\newcommand{\Foot}{\mathsf F}

\newcommand{\class}[1]{\ensuremath{\mathsf{#1}}\xspace}
\mathchardef\mhyphen="2D %
\newcommand{\QMA}{\class{QMA}}
\newcommand{\BQP}{\class{BQP}}

\definecolor{clkblue}{HTML}{4A7FB5}
\definecolor{posteal}{HTML}{3A9E8F}
\definecolor{swrose}{HTML}{C46B6B}
\definecolor{bustgt}{HTML}{7B6DAD}
\definecolor{memgold}{HTML}{C49B3A}

\newcommand{\bits}{\{0,1\}}

\newcommand{\polylog}{\mathrm{polylog}\,}
\newcommand{\xor}{\oplus}

\newcommand*{\rom}[1]{\expandafter\@slowromancap\romannumeral #1@}

\usepackage{authblk}

\title{\huge Scalable Passive QRAM}
\author[1,2]{Siddhartha Jain}
\author[1]{Alexander M. Dalzell}
\author[1]{Connor T. Hann}

\affil[1]{\small AWS Center for Quantum Computing, Pasadena, CA 91106, USA}
\affil[2]{\small The University of Texas at Austin, TX 78712, USA}

\begin{document}

\maketitle
\begin{abstract}
    We present a blueprint for implementing the physical QRAM operation in a scalable, passive manner, proving its feasibility in principle. Our construction is passive because for an $n$-qubit address over $N = 2^n$ memory bits the energy cost per QRAM query is $O(\log N)$, and to our knowledge it is the first construction to use $\polylog N$ energy. The runtime for each query is $O(\log^2 N)$. It is achieved by evolving a pre-manufactured time-independent 4-local Hamiltonian with $O(N)$ terms where every qubit participates in at most 4 terms.
\end{abstract}
\section{Introduction}
 
The three decades of quantum algorithms research since Shor's factoring algorithm~\cite{Shor97} have coincided with the ascendancy of big data and machine learning. Yet the most promising candidates for practical quantum advantage lie firmly in the \textit{small}-data regime. For example, the input to a 2048-bit RSA factoring instance fits on a single sheet of paper. Mapping the phase diagram of the Fermi--Hubbard model 
\cite{qin2022hubbard} requires only a handful of inputs, such as the interaction strength $U/t$, filling, temperature, and lattice geometry. And computing the ground-state energy of FeMoco, a standard benchmark, involves a Hamiltonian with only about 180,000 unique coefficients (after truncating near-zero terms)~\cite{Berry2019qubitizationof}; tellingly, later improvements in gate count \cite{vonBurg2021computationalCatalysis,lee2021THC,caesura2025fasterQuantumChemistry,low2025fastQuantumSimulation} came precisely from precomputing factorizations that compress and reshape this input.

The cost of quantum data loading constitutes a fundamental headwind relative to classical methods, which access these small inputs with ease. Is quantum computing destined to remain confined to a niche of ``small-data, big-compute'' problems? A key part of the answer, we believe, lies in an asymmetry between classical and quantum random access memory (RAM and QRAM). Classical algorithms can access massive datasets cheaply because memory access is \textit{passive}. Conventional quantum data access, by contrast, is inherently active and incurs a gate cost that grows linearly with input size~\cite{Jaques2025qramsurveycritique}. This disparity in data loading makes big-data applications essentially dead on arrival as candidates for practical quantum advantage.

\paragraph{This work.} Our main contribution is to provide a blueprint for implementing QRAM passively at the physical level. The blueprint departs from the conventional gate-based paradigm of quantum computing and supports an alternative vision for how quantum computers may one day tackle data-intensive computational problems. The construction is hardware agnostic and abstract, but possesses several features such as robustness to noise which suggest that a scalable, passive QRAM is possible at least in principle.

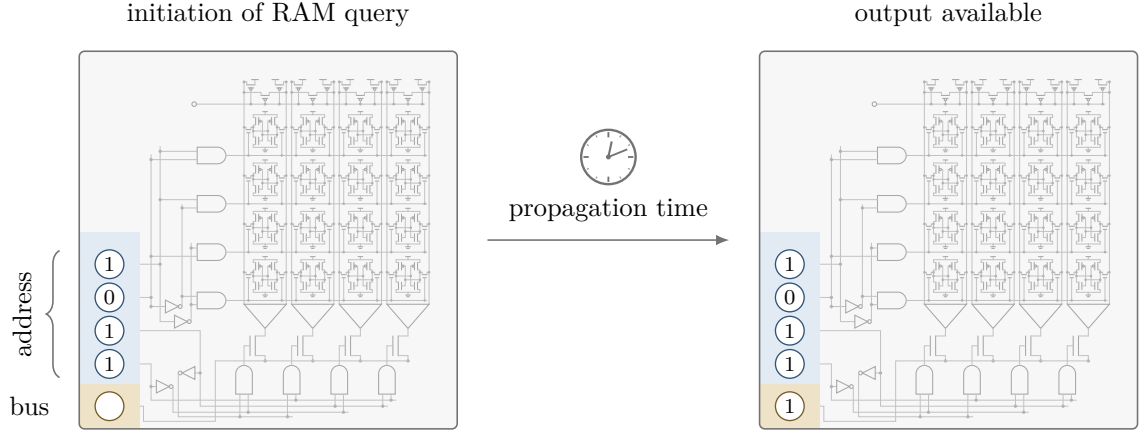
\begin{figure}[!ht]
\centering
\begin{tikzpicture}[
  x=1cm,y=1cm,font=\small,
  wire/.style={draw=black!24,line width=.26pt,line cap=round},
  dev/.style={draw=black!34,line width=.4pt,line cap=round},
  jn/.style={fill=black!30},
  pin/.style={circle,draw=clkblue!65!black,line width=.5pt,fill=white,
              minimum size=3.8mm,inner sep=0pt,font=\scriptsize},
  outpin/.style={circle,draw=memgold!60!black,line width=.5pt,fill=white,
                 minimum size=3.8mm,inner sep=0pt,font=\scriptsize},
]

\newcommand{\andgate}[3]{%
  \begin{scope}[shift={(#1,#2)},rotate=#3]
    \draw[dev] (0,-0.11) -- (0,0.11) -- (0.27,0.11)
      arc[start angle=90,end angle=-90,radius=0.11] -- cycle;
  \end{scope}%
}
\newcommand{\invgate}[3]{%
  \begin{scope}[shift={(#1,#2)},rotate=#3]
    \draw[dev] (0,-0.085) -- (0,0.085) -- (0.17,0) -- cycle;
    \draw[dev] (0.195,0) circle (0.025);
  \end{scope}%
}

\newcommand{\sramcell}[2]{%
  \begin{scope}[shift={(#1,#2)},scale=0.9]
    \draw[dev] (0.17,0.16) -- (0.17,0.27);   %
    \draw[dev] (0.215,0.16) -- (0.215,0.27); %
    \draw[wire] (0.17,0.16) -- (0.135,0.16) -- (0.135,0.12);
    \draw[wire] (0.17,0.27) -- (0.135,0.27) -- (0.135,0.38);
    \draw[wire] (0.215,0.215) -- (0.27,0.215);
    \draw[dev] (0.17,0.45) -- (0.17,0.56);   %
    \draw[dev] (0.215,0.45) -- (0.215,0.56); %
    \draw[wire] (0.17,0.56) -- (0.135,0.56) -- (0.135,0.60);
    \draw[wire] (0.17,0.45) -- (0.135,0.45) -- (0.135,0.38);
    \draw[wire] (0.251,0.505) -- (0.27,0.505);
    \draw[dev] (0.235,0.505) circle (0.016); %
    \draw[dev] (0.45,0.16) -- (0.45,0.27);   %
    \draw[dev] (0.405,0.16) -- (0.405,0.27); %
    \draw[wire] (0.45,0.16) -- (0.485,0.16) -- (0.485,0.12);
    \draw[wire] (0.45,0.27) -- (0.485,0.27) -- (0.485,0.38);
    \draw[wire] (0.405,0.215) -- (0.35,0.215);
    \draw[dev] (0.45,0.45) -- (0.45,0.56);   %
    \draw[dev] (0.405,0.45) -- (0.405,0.56); %
    \draw[wire] (0.45,0.56) -- (0.485,0.56) -- (0.485,0.60);
    \draw[wire] (0.45,0.45) -- (0.485,0.45) -- (0.485,0.38);
    \draw[wire] (0.369,0.505) -- (0.35,0.505);
    \draw[dev] (0.385,0.505) circle (0.016); %
    \draw[wire] (0.27,0.215) -- (0.27,0.505);
    \draw[wire] (0.35,0.215) -- (0.35,0.505);
    \draw[wire] (0.27,0.305) -- (0.485,0.305);
    \draw[wire] (0.35,0.36) -- (0.135,0.36);
    \draw[wire] (0.135,0.12) -- (0.485,0.12) (0.31,0.12) -- (0.31,0.08);
    \draw[dev] (0.275,0.08) -- (0.345,0.08);
    \draw[dev] (0.291,0.055) -- (0.329,0.055);
    \draw[wire] (0.135,0.60) -- (0.485,0.60) (0.31,0.60) -- (0.31,0.645);
    \draw[dev] (0.275,0.645) -- (0.345,0.645);
    \draw[dev] (0.03,0.38) -- (0.10,0.38);   %
    \draw[dev] (0.04,0.335) -- (0.09,0.335); %
    \draw[wire] (0.065,0.335) -- (0.065,0);
    \draw[wire] (0.03,0.38) -- (0.03,0.415) -- (0,0.415);
    \draw[wire] (0.10,0.38) -- (0.10,0.415) -- (0.135,0.415);
    \draw[dev] (0.52,0.38) -- (0.59,0.38);   %
    \draw[dev] (0.53,0.335) -- (0.58,0.335); %
    \draw[wire] (0.555,0.335) -- (0.555,0);
    \draw[wire] (0.59,0.38) -- (0.59,0.415) -- (0.62,0.415);
    \draw[wire] (0.52,0.38) -- (0.52,0.415) -- (0.485,0.415);
    \foreach \p in {(0.485,0.305),(0.135,0.36),(0.27,0.305),(0.35,0.36),
                    (0.135,0.415),(0.485,0.415),(0,0.415),(0.62,0.415),
                    (0.065,0),(0.555,0)}{\fill[jn] \p circle (0.019);}
  \end{scope}%
}

\newcommand{\precharge}[1]{%
  \draw[dev] (#1+0.225,4.415) -- (#1+0.333,4.415);  %
  \draw[dev] (#1+0.243,4.375) -- (#1+0.315,4.375);  %
  \draw[wire] (#1+0.279,4.343) -- (#1+0.279,4.30); %
  \draw[dev] (#1+0.279,4.359) circle (0.016);      %
  \draw[wire] (#1+0.225,4.415) -- (#1+0.225,4.455) -- (#1,4.455);
  \draw[wire] (#1+0.333,4.415) -- (#1+0.333,4.455) -- (#1+0.558,4.455);
  \draw[dev] (#1+0.036,4.555) -- (#1+0.144,4.555);
  \draw[dev] (#1+0.054,4.515) -- (#1+0.126,4.515);
  \draw[wire] (#1+0.09,4.483) -- (#1+0.09,4.30);
  \draw[dev] (#1+0.09,4.499) circle (0.016);
  \draw[wire] (#1+0.036,4.555) -- (#1+0.036,4.595) -- (#1,4.595);
  \draw[wire] (#1+0.144,4.555) -- (#1+0.144,4.625);
  \draw[dev] (#1+0.099,4.625) -- (#1+0.189,4.625);  %
  \draw[dev] (#1+0.414,4.555) -- (#1+0.522,4.555);
  \draw[dev] (#1+0.432,4.515) -- (#1+0.504,4.515);
  \draw[wire] (#1+0.468,4.483) -- (#1+0.468,4.30);
  \draw[dev] (#1+0.468,4.499) circle (0.016);
  \draw[wire] (#1+0.522,4.555) -- (#1+0.522,4.595) -- (#1+0.558,4.595);
  \draw[wire] (#1+0.414,4.555) -- (#1+0.414,4.625);
  \draw[dev] (#1+0.369,4.625) -- (#1+0.459,4.625);  %
  \foreach \p in {(#1,4.455),(#1+0.558,4.455),(#1,4.595),(#1+0.558,4.595),
                  (#1+0.279,4.30),(#1+0.09,4.30),(#1+0.468,4.30)}{%
    \fill[jn] \p circle (0.019);}
}

\newcommand{\colsel}[1]{%
  \draw[dev] (#1,1.65) -- (#1+0.558,1.65) -- (#1+0.279,1.33) -- cycle;
  \draw[wire] (#1+0.279,1.33) -- (#1+0.279,1.23);
  \draw[dev] (#1+0.149,0.97) -- (#1+0.149,1.23);
  \draw[dev] (#1+0.079,0.97) -- (#1+0.079,1.23);
  \draw[wire] (#1+0.149,1.23) -- (#1+0.279,1.23);
  \draw[wire] (#1+0.149,0.97) -- (#1+0.279,0.97) -- (#1+0.279,0.90);
  \draw[wire] (#1+0.079,1.10) -- (#1,1.10) -- (#1,0.84);
  \fill[jn] (#1+0.279,0.90) circle (0.022);
}

\newcommand{\ramcircuit}{%
  \draw[wire] (0.80,1.74) -- (0.95,1.74);                  %
  \draw[wire] (0.95,1.60) -- (0.95,3.63);                  %
  \draw[wire] (0.80,2.18) -- (1.07,2.18);                  %
  \draw[wire] (1.07,1.42) -- (1.07,3.74);                  %
  \draw[wire] (0.95,1.62) -- (1.14,1.62);                  %
  \invgate{1.14}{1.62}{0}
  \draw[wire] (1.36,1.62) -- (1.36,2.97);                  %
  \draw[wire] (1.07,1.42) -- (1.26,1.42);                  %
  \invgate{1.26}{1.42}{0}
  \draw[wire] (1.48,1.42) -- (1.48,2.46);                  %
  \fill[jn] (0.95,1.74) circle (0.022);
  \fill[jn] (1.07,2.18) circle (0.022);
  \fill[jn] (0.95,1.62) circle (0.022);
  \foreach \wy/\ra/\rb in {1.70/1.36/1.48, 2.34/1.48/0.95,
                           2.98/1.07/1.36, 3.62/1.07/0.95}{
    \andgate{1.56}{\wy}{0}
    \draw[wire] (1.56,\wy+0.055) -- (\ra,\wy+0.055);
    \draw[wire] (1.56,\wy-0.055) -- (\rb,\wy-0.055);
    \fill[jn] (\ra,\wy+0.055) circle (0.022);
    \fill[jn] (\rb,\wy-0.055) circle (0.022);
    \draw[wire] (1.94,\wy) -- (4.60,\wy);
  }
  \draw[wire] (0.80,0.86) -- (0.95,0.86) -- (0.95,0.38) -- (4.19,0.38); %
  \draw[wire] (0.95,0.56) -- (1.02,0.56);                  %
  \invgate{1.02}{0.56}{0}
  \draw[wire] (1.24,0.56) -- (1.24,0.22) -- (3.56,0.22);   %
  \draw[wire] (0.80,1.30) -- (1.60,1.30) -- (1.60,0.30) -- (4.08,0.30); %
  \draw[wire] (1.60,0.74) -- (1.53,0.74);                  %
  \invgate{1.53}{0.74}{180}
  \draw[wire] (1.31,0.74) -- (1.31,0.16) -- (2.82,0.16);   %
  \fill[jn] (0.95,0.56) circle (0.022);
  \fill[jn] (1.60,0.74) circle (0.022);
  \foreach \bx/\ra/\rb in {2.18/0.16/0.22, 2.81/0.16/0.38,
                           3.44/0.30/0.22, 4.07/0.30/0.38}{
    \andgate{\bx}{0.46}{90}
    \draw[wire] (\bx-0.055,0.46) -- (\bx-0.055,\ra);
    \draw[wire] (\bx+0.055,0.46) -- (\bx+0.055,\rb);
    \fill[jn] (\bx-0.055,\ra) circle (0.022);
    \fill[jn] (\bx+0.055,\rb) circle (0.022);
  }
  \draw[wire] (1.56,4.30) -- (4.57,4.30);
  \draw[dev] (1.52,4.30) circle (0.03);
  \foreach \bx in {2.18,2.81,3.44,4.07}{
    \draw[wire] (\bx,4.595) -- (\bx,1.65);
    \draw[wire] (\bx+0.558,4.595) -- (\bx+0.558,1.65);
    \precharge{\bx}
    \foreach \wy in {1.70,2.34,2.98,3.62}{\sramcell{\bx}{\wy}}
    \colsel{\bx}
  }
  \draw[wire] (4.349,0.90) -- (1.80,0.90) -- (1.80,0.10)
           -- (0.85,0.10) -- (0.85,0.30) -- (0.80,0.30);
}

\newcommand{\chippanel}[4]{%
  \begin{scope}[shift={(#1,0)}]
    \begin{scope}
      \clip[rounded corners=2pt] (0,0) rectangle (5,5);
      \fill[black!3] (0,0) rectangle (5,5);
      \ramcircuit
      \fill[memgold!28] (0,0) rectangle (0.8,0.6);   %
      \fill[clkblue!16] (0,0.6) rectangle (0.8,2.6); %
    \end{scope}
    \foreach \bit [count=\k from 0] in {#3}{
      \pgfmathsetmacro{\yy}{0.86+.44*\k}
      \node[pin] at (0.4,\yy) {\bit};
    }
    \node[outpin] at (0.4,0.3) {#4};
    \draw[rounded corners=2pt,draw=black!55,line width=.6pt] (0,0) rectangle (5,5);
    \node[anchor=south,align=center] at (2.5,5.2) {#2};
  \end{scope}%
}

\chippanel{0}{initiation of RAM query}{1,1,0,1}{}
\chippanel{9}{output available}{1,1,0,1}{1}

\draw[black!55,line width=.6pt,decorate,decoration={brace,amplitude=5pt}]
  (-0.26,0.68) -- (-0.26,2.36);
\node[rotate=90,anchor=south,font=\small] at (-0.50,1.52) {address};
\node[anchor=east,font=\small] at (-0.26,0.30) {bus};

\draw[-{Latex[length=5pt,width=4pt]},black!55,line width=.6pt] (5.4,2.5) -- (8.6,2.5);
\node[anchor=south,font=\small] at (7.0,2.62) {propagation time};

\begin{scope}[shift={(7.0,3.60)}]
  \draw[black!55,line width=.9pt] (0,0) circle (0.36);
  \foreach \a in {0,90,180,270}{\draw[black!55,line width=.7pt] (\a:0.26) -- (\a:0.33);}
  \foreach \a in {30,60,120,150,210,240,300,330}{%
    \draw[black!40,line width=.45pt] (\a:0.30) -- (\a:0.33);}
  \draw[black!55,line width=.8pt] (0,0) -- (78:0.21);   %
  \draw[black!55,line width=.7pt] (0,0) -- (22:0.27);   %
  \fill[black!55] (0,0) circle (0.032);
\end{scope}

\end{tikzpicture}%
\caption{Illustration of passive RAM query}
\label{fig:ram-chip}
\end{figure}

\paragraph{Passive RAM.}
It is instructive to recall an abstract picture of how classical RAM works. For a dataset of size $N$, a circuit with $O(N)$ elements is physically printed on the chip at fabrication time. As depicted in \cref{fig:ram-chip}, a query is initiated by setting voltages on $O(\log N)$ address pins, after which only a small fraction of the $O(N)$ elements \emph{switch} in response, autonomously routing the requested data to the output pins through a circuit of only $O(\log N)$ logical depth. In asynchronous memories, the requested data simply appears at the output a fixed propagation time after the address is set---in practice just a handful of nanoseconds. No element requires individual external control, and elements that do not switch dissipate no dynamic energy, so the energy expended by a query scales sublinearly with $N$. In an idealized circuit model, total energy dissipated can be as little as $\polylog N$, though practical memories laid out as two-dimensional arrays switch an entire row of cells on each access.  These two properties---no individual control and negligible energy for idle elements---make it possible to shrink the elements and pack them at enormous density, so memory capacity has grown by many orders of magnitude while the cost of each query has grown only weakly with the size of the data. Accordingly, the standard cost model for classical algorithms treats each memory query as a unit-cost operation, regardless of how much data is stored.

\paragraph{QRAM and its challenges.} QRAM is a theoretical primitive that generalizes the task performed by RAM: it allows rapid access to an arbitrary location in a massive memory, and is augmented with the ability to query this memory in \emph{superposition}. Inspired by RAM, quantum algorithms designers commonly model QRAM as a unit-cost operation. In fact, this ``cheap QRAM'' assumption is required to realize quantum advantage with superposition queries to an arbitrary oracle without a succint implementation~\cite{Grover96,BernsteinVazirani,Kuperberg13,AAForrelation}.

But it remains contentious to make this assumption,
because we do not have a compelling plan for how to implement this superposition access to classical memory in a scalable fashion on quantum computers.
Toward this end, the famous bucket-brigade circuit~\cite{GLL08} makes some progress. It implements QRAM access to memories of size $N$ in depth $\polylog N$ while ``switching'' only $\polylog N$ components, resembling how a RAM query switches only a small fraction of its circuit elements. However, the configuration of switched components  is a superposition spanning the whole device, entangled with the address. Thus, to operate the circuit one still has to apply a quantum gate to actively \emph{check} (via application of a quantum gate) whether each component should be switched, which needs $\Omega(N)$ interventions from classical control and expends $\Omega(N)$ energy. This conclusion is difficult to avoid since any circuit for QRAM requires $\Omega(N)$ gates, and in the dominant paradigm for quantum computing, information processing is done actively by applying pulses or other interventions. 
Assuming cheap access to QRAM was one of the main concerns raised in an essay by Aaronson~\cite{Aaronson2015ReadTF} which asked readers to be careful about claimed quantum speedups for machine learning applications~\cite{HHL09,Childs2009-mk,wiebe2015quantumdeeplearning,LloydMohseniRebentrost2013,RebentrostMohseniLloyd2014,WiebeBraunLloyd2012,GarneroneZanardiLidar2012}. Outside of quantum machine learning, availability of QRAM would also lead to better algorithms for dihedral HSP~\cite{Kuperberg13}.
More generally, it would address the main challenge
in quantum applications~\cite{babbush2025grandchallengequantumapplications}, by enabling us to take algorithms from abstract advantages in oracle-models to ensembles of inputs where we can demonstrate quantum advantage.
In particular, with QRAM we would not require a streaming model to achieve quantum advantage for big-data problems~\cite{zhao2026exponentialquantumadvantageprocessing}.
The challenge of implementing QRAM is well-documented~\cite{Jaques2025qramsurveycritique}, wherein the passive approach to QRAM was defined but considered challenging, presented along with partial impossibility results.
We thus solve the main open question from this survey in \cref{thm:main}, circumventing their impossibility results.

\paragraph{What we mean by passive.} We say that a proposal for QRAM is \emph{passive} if it uses $\polylog N$ energy per query to a memory of size $N$, as opposed to the $\Omega(N)$ energy it would naively take to implement the full circuit, following \cite{Jaques2025qramsurveycritique}.
We take the perspective that achieving a passive QRAM requires designing a time-independent Hamiltonian which implements the QRAM query autonomously, as posed in \cite[Section 6.3]{Jaques2025qramsurveycritique}. Since this Hamiltonian ``naturally'' implements the routing and read operations, this process consumes no energy. 
The energy consumption thus comes only from initializing and resetting the qubits the Hamiltonian acts on.  Another benefit of this approach is that, since gates are performed autonomously, in principle it removes the need to allocate individual control lines for each component of the device, an important bottleneck in scaling up quantum hardware~\cite{franke2019rent,kennedy2026designoperationwaferscalepackages}.

It's not obvious this circumvents the energy bottleneck because the Hamiltonian must necessarily have a $\Omega(N)$ size workspace, and all those qubits may need to be reset after every query. This also means that the number of \emph{terms} in the Hamiltonian must be $\Omega(N)$. 
Although this means a linear cost at the time of manufacturing, this is an unavoidable one-time cost, and also true of classical RAM. Thus, we only care about the cost per query. We show that our construction only needs $O(\log N)$ qubits to be reset on every query, giving us a passive QRAM. This circumvents the barrier shown by~\cite{Jaques2025qramsurveycritique}, which shows that a Hamiltonian that implements QRAM must have spectral norm $\Omega(N)$. Our construction does indeed have large spectral norm, but it is a quantum many-body system with $\Omega(N)$ modes and only $\polylog N$ excitations which hop around in those modes, which still results in $\polylog N$ energy consumption.

\paragraph{What we mean by scalable.} For the QRAM Hamiltonian to be manufactured at scale, the construction should break down into a finite number of hardware \emph{primitives} that repeat many times to build the Hamiltonian. Further, the complexity of the primitive should scale mildly with $N$, which we model as $\polylog N$. For example, this means that as $N$ grows, we should only need to implement the primitives with imprecision $1/\polylog N$ as opposed to $1/N$. Our Hamiltonian construction is given in terms of these primitives, and our noise resilience results confirm that these primitives only need to be manufactured with mildly small imprecision. At a high level, our Hamiltonian is noise resilient because the primitives are arranged as a binary tree  as in the bucket-brigade protocol~\cite{Giovannetti2008Architectures, HannNoise21}. Thus, any errors in one branch of the tree do not affect the other branches in a meaningful way.

We are thus ready to state our main result. We state an informal result which is a consequence of our construction in \cref{thm:exact qram} and the statement of \cref{thm:intro-noise}.

\begin{mainresult}(Informal)
\label{thm:main}
    An asymptotically scalable passive QRAM which evolves a time-independent Hamiltonian is feasible in principle. For $n$-qubit queries over $N = 2^n$ bits, the energy cost per query is $O(\log N)$ and the runtime is $O(\log^2 N)$. Moreover, queries can be executed with constant error even if there is $1/\polylog N$ imprecision in each component.
\end{mainresult}

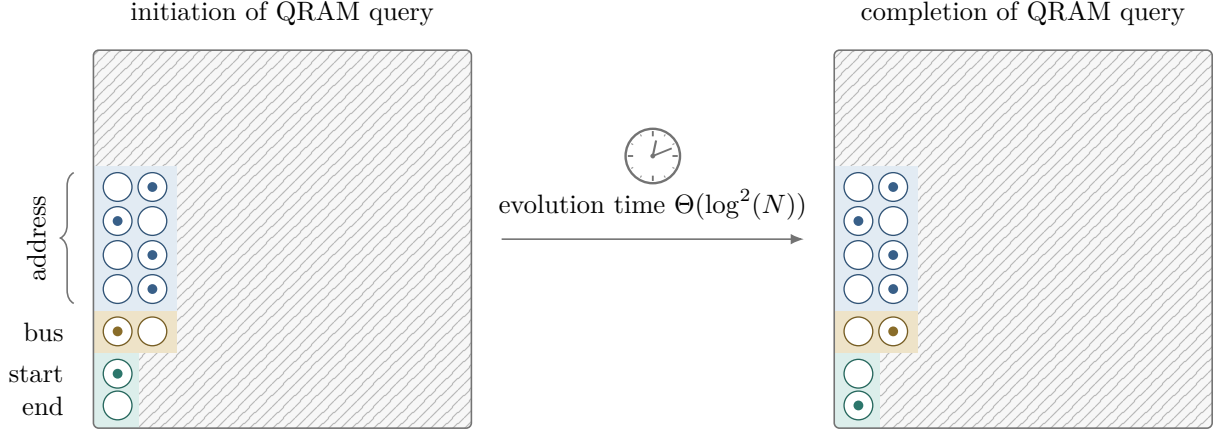
\begin{figure}[!ht]
\centering
\begin{tikzpicture}[
  x=1cm,y=1cm,font=\small,
  addr/.style={circle,draw=clkblue!65!black,line width=.5pt,fill=white,
               minimum size=3.8mm,inner sep=0pt,font=\scriptsize},
  bus/.style={circle,draw=memgold!60!black,line width=.5pt,fill=white,
              minimum size=3.8mm,inner sep=0pt,font=\scriptsize},
  aux/.style={circle,draw=posteal!65!black,line width=.5pt,fill=white,
              minimum size=3.8mm,inner sep=0pt,font=\scriptsize},
]

\newcommand{\qrampanel}[4]{%
  \begin{scope}[shift={(#1,0)}]
    \begin{scope}
      \clip[rounded corners=2pt] (0,0) rectangle (5,5);
      \fill[black!3] (0,0) rectangle (5,5);
      \fill[pattern={Lines[angle=45,distance=3pt,line width=.4pt]},
            pattern color=black!28] (0,0) rectangle (5,5);
      \fill[posteal!18] (0,0)    rectangle (0.60,1.00); %
      \fill[memgold!28] (0,1.00) rectangle (1.1,1.56);  %
      \fill[clkblue!16] (0,1.56) rectangle (1.1,3.47);  %
    \end{scope}
    \node[aux] at (0.32,0.30) {};
    \node[aux] at (0.32,0.72) {};
    \fill[posteal!75!black] (0.32,#4) circle (0.065);
    \node[bus] at (0.32,1.28) {};
    \node[bus] at (0.78,1.28) {};
    \fill[memgold!70!black] (#3,1.28) circle (0.065);
    \foreach \yy/\ax in {1.84/0.78, 2.29/0.78, 2.74/0.32, 3.19/0.78}{
      \node[addr] at (0.32,\yy) {};
      \node[addr] at (0.78,\yy) {};
      \fill[clkblue!75!black] (\ax,\yy) circle (0.065);
    }
    \draw[rounded corners=2pt,draw=black!55,line width=.6pt] (0,0) rectangle (5,5);
    \node[anchor=south,align=center] at (2.5,5.2) {#2};
  \end{scope}%
}

\qrampanel{0}{initiation of QRAM query}{0.32}{0.72}
\qrampanel{9.8}{completion of QRAM query}{0.78}{0.30}

\draw[black!55,line width=.6pt,decorate,decoration={brace,amplitude=5pt}]
  (-0.26,1.65) -- (-0.26,3.38);
\node[rotate=90,anchor=south,font=\small] at (-0.50,2.515) {address};
\node[anchor=east,font=\small] at (-0.26,1.28) {bus};
\node[anchor=east,font=\small] at (-0.26,0.72) {start};
\node[anchor=east,font=\small] at (-0.26,0.30) {end};

\draw[-{Latex[length=5pt,width=4pt]},black!55,line width=.6pt] (5.4,2.5) -- (9.4,2.5);
\node[anchor=south,font=\small] at (7.4,2.62) {evolution time $\Theta(\log^2(N))$};

\begin{scope}[shift={(7.4,3.60)}]
  \draw[black!55,line width=.9pt] (0,0) circle (0.36);
  \foreach \a in {0,90,180,270}{\draw[black!55,line width=.7pt] (\a:0.26) -- (\a:0.33);}
  \foreach \a in {30,60,120,150,210,240,300,330}{%
    \draw[black!40,line width=.45pt] (\a:0.30) -- (\a:0.33);}
  \draw[black!55,line width=.8pt] (0,0) -- (78:0.21);
  \draw[black!55,line width=.7pt] (0,0) -- (22:0.27);
  \fill[black!55] (0,0) circle (0.032);
\end{scope}

\end{tikzpicture}%
\caption{Illustration of passive QRAM query accomplished by our construction.  To match the notation we adopt in \cref{sec:hamiltonian}, we depict activated qubits in $\ket{1}$ as $\bullet$ and $\ket{0}$ as empty. The Hamiltonian is hard-coded into the device at time of fabrication (here indicated by the hatching), and is always on.
We show that our construction has robustness to imperfections in the manufacture of the chip and to thermal noise in the initial state of the $O(N)$ auxiliary qubits.}
\label{fig:qram-chip}
\end{figure}

Our construction is visualized in \cref{fig:qram-chip}. The address and bus qubits are encoded in unary (dual-rail), and can be initialized into an arbitrary superposition. The query is initiated by exciting the auxiliary qubit labeled ``start.'' At exactly a specific $\Theta(\log^2 N)$ amount of time later, the query completes: the address and bus register are returned having picked up the corresponding data bit in the bus, and the excitation that began in the start register is now in the register labeled ``end.'' Active control only acts on the $O(\log N)$ qubits depicted as circles in the figure: at the beginning to initialize and set off the evolution, and at the end to retrieve the answer and halt the evolution.

\paragraph{Locality and degree.}
When designing a Hamiltonian, we need to assign constraints to make it feasible to physically manufacture it. We optimize for two fundamental parameters of the Hamiltonian: its locality (how many qubits participate in each term), and the degree of its interaction hypergraph induced by the terms (how many terms a qubit participates in).

We are able to achieve very simple 4-local interactions. While our construction is agnostic to any specific platform, this is particularly compelling in the context of superconducting qubits because the famous Josephson junction~\cite{Josephson1962,MDCa,MDCb,MDCc} can create fairly general 4-body interactions with standard components.
Having numerous Hamiltonian terms which are 5-local or worse would introduce significant engineering complexity.
Along with the few-body interactions, we also have maximum qubit degree 4. This is useful because platforms like superconducting circuits have practical limits on how many neighbors a given qubit can interact with, for example because each additional coupling adds capacitance and occupies space around the qubit, and increases frequency crowding and crosstalk.

\paragraph{Integration into a fault-tolerant quantum algorithm.} 

Furthermore, a noisy physical QRAM operation implemented by this blueprint can be useful within a fault-tolerant quantum computation using the distillation--teleportation protocol of~\cite{Dalzell2025DistillationTeleportationFTQRAM}.
The logical QRAM operation on a memory of size $N$ can be accomplished up to arbitrary $1/\polylog N$ error with $\polylog N$ queries to our blueprint passive QRAM device and $\polylog N$ other fault-tolerant gates. Thus, our work and \cite{Dalzell2025DistillationTeleportationFTQRAM} together show that logical QRAM can incur only $\polylog N$ total quantum cost (in terms of energy). However, achieving fault tolerance with the distillation--teleportation approach of \cite{Dalzell2025DistillationTeleportationFTQRAM} still incurs at least $\Omega(N)$ purely classical cost for adaptive memory updates.  Thus, we do not achieve ``strongly passive'' QRAM,  as defined in \cite{Jaques2025qramsurveycritique, zhang2026bitsqubitstheorypractice}, which requires $o(N)$ complexity across the whole stack, including classical computation. We have also therefore \textit{not} yet fully justified the cheap QRAM assumption used in the analysis of quantum algorithms. Nevertheless, our construction supports the conclusion that physical QRAM need not be a barrier in this research program. 

\paragraph{Outline and Outlook.} We believe this is a reasonable proposal for passive physical QRAM. It formally avoids the linear energy cost of any active approach, and it mirrors the operation of passive classical RAM as in \cref{fig:ram-chip}. With low enough imprecision in fabrication, and a suitable layout of the qubits and interaction terms, this blueprint could be realized on any platform that would support such a passive Hamiltonian evolution. Candidates include superconducting and silicon photonic architectures.  In contrast,  neutral atom and trapped ion architectures appear less compatible since active energy input is required for trapping in proportion to the number of atoms/ions.

Our approach differs from previous experimental proposals or implementations~\cite{GLL08,Jaques2025qramsurveycritique,miaoetalqrouter,Shen2026BucketBrigade} in the crucial aspect that we first construct an abstract Hamiltonian encoding the circuit which is reminiscent of the Feynman-Kitaev clock constructions to synchronize the time evolution of different gates.
Previous QRAM proposals did not perform routing passively, and hence consumed energy linear in memory size.
The technical novelty in our construction is to specialize the clock construction to QRAM which allows us to optimize many parameters.
We state our results in \cref{sec:our_results} below. In \cref{sec:hamiltonian} we focus on defining a time-independent Hamiltonian satisfying the conditions of \cref{thm:exact qram}. 
\Cref{sec:experiment} discusses a proposal of how one would manufacture this Hamiltonian with our current knowledge of hardware, although our construction is agnostic to any platform or hardware type. In \cref{sec:noise} we focus on how to avoid resetting an exponential number of qubits on every query, and how to deal with imprecision during manufacturing and runtime. 
Combining these results, we have a proposal for passive QRAM. There is work to be done to figure out the best platform, the best embedding of the Hamiltonian into 2D or 3D space, to figure out how to engineer the 4-body couplings needed for this Hamiltonian at scale, and to deal with the long range 2-body interactions as explained in \cref{sec:experiment}. Note that via lightcone arguments \cite{wang2024fundamental} the long-range interactions are unavoidable, indeed this is the case even in classical RAM. But we remain hopeful because in practice the speed of light is not a bottleneck.

For silicon-based platforms, there is also the bottleneck of the largest wafer that could be manufactured with existing technology.  There is a practical limit of how close we could put two Hamiltonian terms but superconducting architectures have demonstrated sites less than 1 micron apart~\cite{Macha2014QuantumMetamaterial}, and recent works on spin qubits have demonstrated less than 100 nanometers~\cite{Philips2022SixQubitSilicon,HRL2026SiliconQPU}.
The largest wafers manufactured with an industrial process is about 300mm in diameter~\cite{Van_Damme2024-eq}. This lets us get up to $n = 32$ conservatively, corresponding to 500 MiB of QRAM. In classical hardware, memory sizes have grown because components have miniaturized and become extremely dense, and we may hope this will be repeated for quantum hardware as well. 
For $N=2^{32}$ and assuming operating frequencies typical for superconducting qubits (around 5 GHz), a back of the envelope calculation suggests that a temperature of 14mK is sufficient to keep enough of system in their ground state at initialization for the overall query error to be less than 10\%. 
This is around the temperature at which dilution refrigerators are operated for superconducting hardware.
This makes QRAM with memory $\approx 500$MB conceivable, acknowledging several theoretical (such as exact qubit layout) and experimental (such as long-range Hamiltonian terms) ingredients still remain to be specified.
We leave the details of choosing the best platform and fabricating the passive Hamiltonian for future work and 
hope that this encourages interest in the realization of scalable passive QRAM.

\subsection{Related work}
Prior work deals with three questions that are easy to blur together: (i) how to route a
coherent address on hardware, (ii) how to make the query autonomous, and (iii) how errors scale with the
memory size.

\paragraph{Routing and hardware.}
Giovannetti, Lloyd, and Maccone introduced the fanout and bucket-brigade
architectures~\cite{GLL08,Giovannetti2008Architectures}. A bucket-brigade
device contains $\Theta(N)$ routers arranged in a depth-$n$ binary tree, but only the $n$ routers corresponding to a path from root to leaf store
the address in each computational-basis branch. Their concrete cavity scheme
still loads the address, routes the bus, and erases the path with a timed
sequence of pulses. This is inherently an active implementation of QRAM and requires $\Omega(N)$ energy to implement.
More recent superconducting experiments have implemented an addressable router~\cite{miaoetalqrouter} and circuit-based bucket-brigade
memories with four and eight data bits~\cite{Shen2026BucketBrigade}. These
works test active routing primitives. They do not implement a query by free
evolution under one fixed Hamiltonian. Our construction is also \emph{passive} in this sense: after the inputs are prepared and the query begins, the
couplings stay fixed until readout.

\paragraph{Hamiltonian formulation.}
Jaques and Rattew sharpen the distinction between active and passive
QRAM~\cite{Jaques2025qramsurveycritique}. An external controller cannot learn
which branches support a coherent address, so an active protocol must be
ready to control all $\Theta(N)$ routers.
This motivates a time-independent Hamiltonian construction which implements QRAM queries
(or at least one with few active interventions).
But so far, there has been no satisying solution which builds a device whose couplers are built
once and whose idle state is perturbed to begin a query. A time-independent Hamiltonian approach was considered in Asaka, Sakai, and Yahagi~\cite{quantumwalkers1,quantumwalkers2}. But their proposal is not known to be scalable or passive. We do not know how to prove noise resilience gaurantees for their proposal.
Bohac gives another time-independent QRAM
Hamiltonian but it has interaction degree $\Theta(N)$~\cite{Bohac2025UniversalQRAM}.

Feynman, Kitaev, and later works give circuit-to-Hamiltonian constructions~\cite{feynman1986quantum,
kitaev1997quantum,Nagaj2010}. The space-time construction of Breuckmann and
Terhal also encodes the parallelism of the circuit in a local
Hamiltonian~\cite{BB14}. Applied directly to QRAM, a sequential clock follows
the $\Theta(N)$ gates of the routing circuit, while the space-time
construction targets a history state rather than exact transfer to the
output at a prescribed time. Our construction instead uses the tree
structure of bucket brigade. Each basis query sees one path, and perfect
state transfer along that path~\cite{Christandl2004} produces the output
exactly after $O(n^2)$ time.

Our Hamiltonian
may couple distant qubits, so it does not evade the causal bounds for
finite-range interactions proved by Wang et al.~\cite{Wang2024CausalQRAM}.

\paragraph{Noise.}
Hann et al.\ proved local-noise resilience for the bucket-brigade circuit~\cite{HannNoise21}. Their key observation is branch isolation, that is, an error at one router affects only address branches that visit it. They consider a qutrit implementation such that the infidelity of the output of their circuit is
$O(\epsilon\log^2 N)$ with the ideal query state, under their local-channel error model where every router suffers an error with probability $\epsilon$. They also give a
two-query protocol that tolerates arbitrary initialization. We prove a
continuous-time version of branch isolation and bound the channel induced by tracing out the qubits beside the address and the bus
in diamond norm. Since the models are different, neither result subsumes the other.

\subsection{Our results}

\label{sec:our_results}
Let $N=2^n$. We have $N$ classical bits of data encoded by $D=(D_x)_{x\in\bits^n}$, stored in memory qubits $M_x$ as $\ket{D_x}$.  The address register is $A=(A_1,\ldots,A_n)$ and the bus is one qubit $B$.  The target QRAM unitary $U$ does not modify the classical data $D$ and can be written as 
\begin{align}
    U = \sum_D U_D \otimes \ket{D}\bra{D}_M\,,\label{eq:qram_U}
\end{align}
where
\begin{equation}
    U_D \ket a_A\ket b_B
    =\ket a_A\ket{b\xor D_a}_B,
    \qquad a\in\bits^n,\; b\in\bits .
    \label{eq:xor-qram}
\end{equation}
We construct a Hamiltonian where each term involves at most 4 qubits and each qubit participates in at most 4 terms, which implements this unitary. The Hamiltonian acts on $O(N)$ auxiliary qubits in addition to the registers specified in \cref{eq:xor-qram}. This Hamiltonian is independent of $D$ and only requires the memory register as an initial state. To initiate a query, we set the auxiliary register to $\ket{\blacktriangleright}_{\textrm{aux}}$ which has a special `begin' qubit in $\ket{1}$ and all other qubits in $\ket{0}$. Then it evolves for a fixed $\poly(n)$ time. 

\begin{theorem}[Passive Hamiltonian QRAM]
\label{thm:exact qram}
For every $n\ge 1$ there is a time-independent, data-independent Hamiltonian $H_n$ acting on $O(N)$ physical qubits such that,
\begin{enumerate}[label=(\roman*),leftmargin=2.2em,itemsep=0.15em]
    \item $H_n$ is $4$-local,
    \item its interaction hypergraph has degree 4,
    \item $e^{-iH_n\tau} \ket{a}_A \ket{b}_B \ket{D}_M \ket{\blacktriangleright}_{\textrm{aux}}$ at $\tau=O(n^2) = O(\log^2 N)$ implements $U$ from~\eqref{eq:qram_U} up to a global phase independent of $a,b,D$ after tracing out the auxiliary qubits.
\end{enumerate}
\end{theorem}

Note that by a lightcone argument, this Hamiltonian cannot be geometrically local in 2D or 3D while having a $\poly(n)$ runtime \cite{wang2024fundamental}. But we will show how to realize all the geometrically nonlocal terms as simple 2-body interactions, which we expect will be scalable. We now make three claims about the resilience of $H_n$ to noise. We assume that the $M$ register storing the classical data $\ket{D}$ does not experience any noise and does not support superpositions. The ideal QRAM channel $\mathcal{U}$ acts on such states as 
\begin{align}
    \mathcal{U}[\rho_{AB} \otimes \ket{D}\bra{D}] = \mathcal{U}_D[\rho_{AB}] \otimes \ket{D}\bra{D}\,,
\end{align}
where $\mathcal{U}_D[\cdot] = U_D[\cdot]U_D^\dagger$. The channel $\mathcal{E}$ realized when noise is included is defined similarly. 
\begin{align}
    \mathcal{E}[\rho_{AB} \otimes \ket{D}\bra{D}_M] = \mathcal{E}_D[\rho_{AB}] \otimes \ket{D}\bra{D}\,,
\end{align}
where $\mathcal{E}_D$ is the noisy channel conditional on $D$.

\begin{theorem}[Noise resilience guarantee]
\label{thm:intro-noise}

For the Hamiltonian $H_n$ in \cref{thm:exact qram} augmented with a 1-local pinning Hamiltonain $P_n$. Then for one query ($m=1$),
\begin{enumerate}[label=(\roman*),leftmargin=2.2em,itemsep=0.15em]
\item {\bf Thermal initialization.}
Prepare the $O(\log N)$ address qubits and the bus cleanly, and initialize the remaining qubits thermally from $K_n = H_n + P_n$. Then $\|\mathcal E_{D}-\mathcal U_D\|_\diamond\le \varepsilon$ once $\beta\ge \Omega(\log \log N + \log (1/\varepsilon))$.

Here $\beta$ is the inverse temperature and $g$ is the effective energy gap of a single qubit. At fixed gap $g$ and $\varepsilon = 1/\polylog N$ this is a temperature $T=O(1/\log\log N)$.

\item {\bf Sensitivity to Hamiltonian terms.} Suppose the realized Hamiltonian is
$\widetilde{K}_n=K_n+\sum_\alpha E_\alpha$. Let $\lambda_{\rm err}$ be the maximum sum of spectral norms of the erroneous terms acting on a single qubit. Then $\|\mathcal E_{D}-\mathcal U_D\|_\diamond\le \varepsilon$ if $\lambda_{\rm err} = O(\varepsilon/\log^4 N)$.
\end{enumerate}

These results extend to arbitrary $m$ for $\beta\ge \Omega(\log \log N + \log (m/\varepsilon))$ and $\lambda_{\rm err}  = O(\varepsilon/m^2 \log^4 N)$.
\end{theorem}

We prove (i) in \cref{sec:thermal} and (ii) in \cref{sec:imperfect}. Note that (i) is important for passive QRAM, and (ii) is important for scalable QRAM. We also extend the proof of (ii) to show resilience against continuous Markovian noise, formalized as a Lindbladian~\cite{GKS1976,Lindblad1976}, in \cref{sec:dynamical}.

\subsection{Open questions}

\begin{itemize}
    \item Can we achieve passive QRAM with better tradeoff for parameters? The relevant axes are the runtime, the energy consumption, the sensitivity to Hamiltonian terms, and the resilience to continuous noise. While we achieve $\polylog N$ scaling on all axes, we note that we have not optimized the scaling in favor of more general results. It is possible there is a better combination of tradeoffs, and we expect that for most platforms the noise resilience will be the bottleneck.
    \item Our construction could conceivably be implemented on multiple physical platforms. What are the tradeoffs between different platforms and which ones are the most promising?
    \item Our proposal can be integrated into a fault-tolerant algorithm using the distillation-teleportation protocol of \cite{Dalzell2025DistillationTeleportationFTQRAM}. While this protocol has $\polylog N$ quantum overhead, it has $\Omega(N)$ overhead from classical control. Is there a different protocol which can get $\polylog N$ overhead for both?
\end{itemize}

\subsection*{Acknowledgements and tool disclosure}

We thank Andr{\'a}s Gily{\'e}n, Catherine Leroux, and Yuan Su for hepful discussions regarding this project. We thank Scott Aaronson for useful discussions about the exposition of this work.
SJ was partially supported by the Amazon AI Fellowship.

The authors proposed the approach to passive QRAM of this work and identified why existing constructions do not achieve it. ChatGPT 5.5 suggested the use of perfect state transfer for the Hamiltonian construction in a conversation. The authors then iterated on the construction to improve parameters, simplify the presentation, and prove better noise resilience gaurantees. Both ChatGPT and Claude then assisted in drafting the proofs and proofreading. The authors assume responsibility for the correctness of the proof and its exposition.

\section{Technical overview}
We first discuss the construction of the time-independent Hamiltonian whose evolution implements \cref{eq:qram_U}. Our goals will be to optimize the locality of the Hamiltonian (how many qubits participate in each term), the degree of the interaction hypergraph induced by the terms (how many terms a qubit participates in), and the evolution time for getting good overlap with the QRAM state, which we want to be $\poly(n)$.

Later, we discuss how we can avoid a reset of exponentially many qubits on every query and active error correction to make the proposal truly passive. Put together, our results resolve the question of how to build passive QRAM (in principle) despite the no-go results proved in a recent survey on QRAM~\cite{Jaques2025qramsurveycritique}. In their survey, they showed that any Hamiltonian implementation of QRAM must have spectral norm which scales as $\Omega(N)$, which they posed as a barrier to implementing passive QRAM. Our Hamiltonian construction indeed has this property, but it only consumes $O(\log N)$ energy. One explanation of this is that during evolution it stays in a low-energy subspace of the entire Hilbert space. 

\subsection{Hamiltonian construction}
Our starting point for this construction is the original bucket-brigade protocol~\cite{GLL08}. The address bits descend through a branching tree one by one and set a path to the chosen memory cell.
The bus follows that path, reads the stored value, returns. Then the path resets or uncomputes. This is visualized in \cref{fig:bucket-brigade}.
Our aim is to encode this computation in a $O(1)$-local, $O(1)$-degree time-independent Hamiltonian where a $\poly(n)$ time evolution gives us a good overlap with the query output state. We can always use the Feynman-Kitaev (FK) circuit-to-Hamiltonian construction~\cite{feynman1986quantum,kitaev1997quantum}, which gives us a 5-local Hamiltonian, but this has two problems for us. First, the QRAM circuit has \(2^n\poly(n)\) gates, so a sequential clock follows an exponentially large circuit and can have large degree. Second, the usual history-state constructions do not put the final state at a known time with unit fidelity. Our construction uses the tree structure to solve both problems.

Note that there are few circuit-to-Hamiltonian constructions in literature which try to incorporate the parallelization in a circuit into the Hamiltonian construction, most notably the Space-time circuit-to-Hamiltonian
construction~\cite{BB14}. Another example, which deviates strongly from the FK framework, is the tensor network approach of Anshu, Breuckmann, and Nguyen \cite{anshu2024circuitToHamiltonian}. But to the best of our knowledge, none of them give Hamiltonian which outputs a state with good overlap with the QRAM state with $\poly(n)$ evolution time and constant degree.

\paragraph{Evolution time \& fidelity.}
Most constructions of circuit-to-Hamiltonian maps are designed so that the history state is the \emph{ground state} of the Hamiltonian, which is not efficient to passively recover from the Hamiltonian evolution in general. This is not an issue in the setting of proving \QMA-completeness of local Hamiltonian problem, since we just need to run verification given a quantum state which is claimed to be the ground state. But for our purposes we want small evolution time (efficient in the sense of \BQP) and high fidelity with the output state of the circuit. If one were to try and recover the ground state of the FK Hamiltonian,
the evolution time needed to get good overlap with the history state can be exponential in the input size even for Hamiltonians corresponding to polynomial size circuits.
While the QRAM circuit has depth $\poly(n)$ it has size $2^n \cdot \poly(n)$.

\tikzset{
  tedge/.style={draw=black!22, line width=0.5pt},
  tnode/.style={circle, draw=black!35, fill=black!5, minimum size=5mm,
                inner sep=0pt, font=\scriptsize},
  leaf/.style={rounded corners=1pt, draw=black!35, fill=black!4,
               minimum size=4.2mm, inner sep=0pt, font=\tiny},
  pathon/.style={#1, line width=1.4pt},
  swlab/.style={font=\scriptsize\figbf},
}
\providecommand{\figbf}{}%
\renewcommand{\figbf}{\fontseries{\bfdefault}\selectfont}
\providecommand{\treeEdges}{}%
\renewcommand{\treeEdges}{%
  \draw[tedge] (0,3)--(-2,2) (0,3)--(2,2)
               (-2,2)--(-3,1) (-2,2)--(-1,1) (2,2)--(1,1) (2,2)--(3,1)
               (-3,1)--(-3.5,0) (-3,1)--(-2.5,0) (-1,1)--(-1.5,0) (-1,1)--(-0.5,0)
               (1,1)--(0.5,0) (1,1)--(1.5,0) (3,1)--(2.5,0) (3,1)--(3.5,0);
}
\providecommand{\treeNodes}{}%
\renewcommand{\treeNodes}{%
  \node[tnode] at (0,3) {};
  \node[tnode] at (-2,2) {}; \node[tnode] at (2,2) {};
  \node[tnode] at (-3,1) {}; \node[tnode] at (-1,1) {};
  \node[tnode] at (1,1) {}; \node[tnode] at (3,1) {};
  \foreach \j/\d in {0/0,1/1,2/1,3/0,4/1,5/1,6/0,7/0}
    \node[leaf] at ({\j-3.5},0) {\d};
}

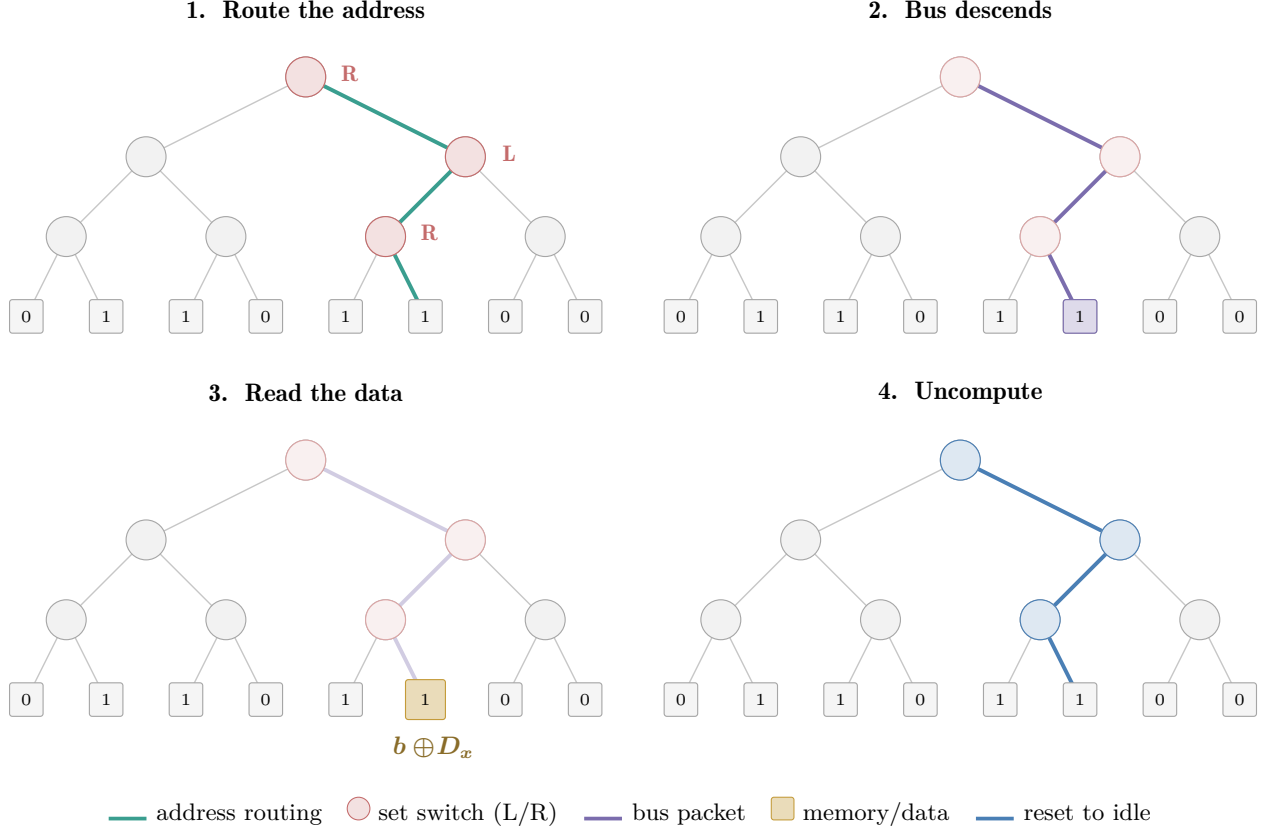
\begin{figure}[ht]
\centering
\resizebox{\textwidth}{!}{%
\begin{tikzpicture}[>={Latex[length=4pt,width=3pt]}]

\begin{scope}[shift={(0,0)}]
  \node[font=\footnotesize\figbf] at (0,3.85) {1.\ \ Route the address};
  \treeEdges
  \draw[pathon=posteal] (0,3)--(2,2);
  \draw[pathon=posteal] (2,2)--(1,1);
  \draw[pathon=posteal] (1,1)--(1.5,0);
  \treeNodes
  \node[tnode, draw=swrose, fill=swrose!20] at (0,3) {};
  \node[tnode, draw=swrose, fill=swrose!20] at (2,2) {};
  \node[tnode, draw=swrose, fill=swrose!20] at (1,1) {};
  \node[swlab, text=swrose] at (0.55,3.05) {R};
  \node[swlab, text=swrose] at (2.55,2.05) {L};
  \node[swlab, text=swrose] at (1.55,1.05) {R};
\end{scope}

\begin{scope}[shift={(8.2,0)}]
  \node[font=\footnotesize\figbf] at (0,3.85) {2.\ \ Bus descends};
  \treeEdges
  \draw[pathon=bustgt] (0,3)--(2,2);
  \draw[pathon=bustgt] (2,2)--(1,1);
  \draw[pathon=bustgt] (1,1)--(1.5,0);
  \treeNodes
  \node[tnode, draw=swrose!60, fill=swrose!10] at (0,3) {};
  \node[tnode, draw=swrose!60, fill=swrose!10] at (2,2) {};
  \node[tnode, draw=swrose!60, fill=swrose!10] at (1,1) {};
  \node[leaf, draw=bustgt, fill=bustgt!25] at (1.5,0) {1};
\end{scope}

\begin{scope}[shift={(0,-4.8)}]
  \node[font=\footnotesize\figbf] at (0,3.85) {3.\ \ Read the data};
  \treeEdges
  \draw[pathon=bustgt!35] (0,3)--(2,2);
  \draw[pathon=bustgt!35] (2,2)--(1,1);
  \draw[pathon=bustgt!35] (1,1)--(1.5,0);
  \treeNodes
  \node[tnode, draw=swrose!60, fill=swrose!10] at (0,3) {};
  \node[tnode, draw=swrose!60, fill=swrose!10] at (2,2) {};
  \node[tnode, draw=swrose!60, fill=swrose!10] at (1,1) {};
  \node[leaf, draw=memgold, fill=memgold!35, minimum size=5mm] at (1.5,0) {1};
  \node[font=\footnotesize\bf, text=memgold!70!black] at (1.6,-0.6) {$b\oplus\!D_x$};
\end{scope}

\begin{scope}[shift={(8.2,-4.8)}]
  \node[font=\footnotesize\figbf] at (0,3.85) {4.\ \ Uncompute};
  \treeEdges
  \draw[pathon=clkblue] (2,2)--(0,3);
  \draw[pathon=clkblue] (1,1)--(2,2);
  \draw[pathon=clkblue] (1.5,0)--(1,1);
  \treeNodes
  \node[tnode, draw=clkblue, fill=clkblue!18] at (0,3) {};
  \node[tnode, draw=clkblue, fill=clkblue!18] at (2,2) {};
  \node[tnode, draw=clkblue, fill=clkblue!18] at (1,1) {};
\end{scope}

\end{tikzpicture}%
}

\medskip
{\footnotesize
\tikz\draw[posteal,line width=1.4pt](0,0)--(0.5,0);~address routing\quad
\tikz\node[circle,draw=swrose,fill=swrose!20,minimum size=3mm,inner sep=0pt]{};~set switch (L/R)\quad
\tikz\draw[bustgt,line width=1.4pt](0,0)--(0.5,0);~bus packet\quad
\tikz\node[rounded corners=1pt,draw=memgold,fill=memgold!35,minimum size=3mm,inner sep=0pt]{};~memory/data\quad
\tikz\draw[clkblue,line width=1.4pt](0,0)--(0.5,0);~reset to idle
}
\caption{The bucket-brigade query on a height-$3$ tree for address $x=101$. (1) A carrier routes
the address down the tree, setting each node's switch; (2) the bus packet descends the routed
path; (3) the target memory is read into the bus at the leaf; (4) the bus ascends and the path is uncomputed back
to idle.
}
\label{fig:bucket-brigade}
\end{figure}

\paragraph{Perfect state transfer.} We achieve both of our goals of $\poly(n)$ time evolution and $O(1)$-degree by using the \emph{tree} structure of the QRAM circuit. We construct a Hamiltonian based on perfect state transfer~\cite{Christandl2004}, which corresponds to a 1D nearest-neighbor Hamiltonian of length $L+1$. Let $\sigma^+_j = \ket{1}\!\bra{0}_j$ and $\sigma^-_j = \ket{0}\!\bra{1}_j$ denote raising and lowering operators on site $j$ (which act trivially on the other $L$ sites). Then the perfect state transfer Hamiltonian can be written as
\begin{align}
H_{\rm PST} &= \sum_{j=1}^{L} J_j (\sigma^+_j \sigma^-_{j+1} + \sigma^+_{j+1} \sigma^-_{j})\label{eq:PST}\\
    J_j&=\frac{\Omega}{2}\sqrt{j(L-j+1)},
    \qquad
    \Omega=\Theta(1/L).
    \label{eq:couplings_intro}
\end{align}
The claim is that if an excitation begins at site $j=1$, then the Hamiltonian evolution is such that at some fixed time $T_{\rm PST}= \Theta(L)$ it has deterministically arrived at site $L+1$:
\begin{align}
    e^{-iH_{\rm PST} T_{\rm PST}}\ket{100\ldots 000} = (-i)^L\ket{000\ldots 001}\,.
\end{align}
This technique has precedence for being used for other Hamiltonian constructions before. It is even used for lower bounds; a famous example is to show that simulating time evolution $e^{-iHt}$ with query access to $H$ requires $\Omega(\norm{H} t)$ queries in the worst-case \cite{Berry2007efficientSparse}. The way this is done is to encode the value of the parity function on $n$ bits in the output of simulating a norm one Hamiltonian $H$ simulated to time $\Theta(n)$. This is done using the perfect state transfer weighting, which gaurantees that after a fixed $\Theta(n)$ time, the output of $e^{-iHt}$ is \emph{exactly} the output of the parity function in a specified register. 
For us, this solves the ``overlap'' problem of Hamiltonian constructions. In standard clock constructions, the history state places only $1/(L+1)$ of its weight on the final step of the computation, preparing it as a ground state incurs polynomial overhead due to the $\Omega(1/L^2)$ spectral gap, and dynamics under uniform couplings disperses the wavepacket rather than delivering it to the end. Here, by contrast, the final step is reached with probability 1 after an evolution time linear in the path length.

\paragraph{QRAM Hamiltonian.} In the QRAM operation, we have an address $x \in \{0,1\}^n$ and a bus qubit $B$, and we want to load the bit $D_x$ at the address $x$ into $B$. It suffices to do this for a single address, by linearity. We will use the perfect state transfer trick for this, where it ensures fidelity 1 with the output state if it follows the path to the correct leaf. 

Perfect state transfer was designed to be applied to linear paths. Our construction navigates back and forth along a binary tree, as in \cref{fig:bucket-brigade}, rather than a linear path. The address $x \in \{0,1\}^n$ determines the sequence of steps along the tree. In the full proof, we formalize this as a directed acyclic graph consisting of many copies of the binary tree, where each $x$ maps to a path in the graph.  It is this path on which we are able to apply perfect state transfer and guarantee that the path is deterministically traversed in time proportional to the length of the path. Different values of $x$ lead to different paths, which route in order to pick up their respective bits $D_x$, but all paths have the same length and thus are completed deterministically at the same time.

In our construction, the relevant path length is $L = \Theta(\log^2(N))$---each step of the path corresponds to one time step in the sequence of configurations depicted in \cref{fig:bucket-brigade}. The length is quadratic because our construction requires placing the $\log N$ address bits at the appropriate nodes sequentially, going back to the root each time. Implementations of bucket-brigade in the circuit model are able to pipeline steps such that $O(\log N)$ depth is achieved \cite{HannNoise21}, but this technique is not obviously compatible with our construction. 

It remains to explain how  to ensure the right path is traversed, given an address setting $\ket{x}$. A key design choice of our construction is that all information is encoded in unary. That is, address states are encoded in pairs of qubits as $\ket{0} = \ket{\bullet \circ}$ and $\ket{1} = \ket{\circ \bullet}$, here adopting the notation we use in the full proof where each two level system is spanned by states $\ket{\circ}$ (unoccupied) and $\ket{\bullet}$ (occupied). For a particular address $x$, the initial state is just a configuration of $O(\log(N))$ excitations $\ket{\bullet}$, and all of the auxiliary qubits are asked to begin unoccupied in $\ket{\circ}$.  The path of length $L$ then corresponds to a sequence of changes to this initial configuration, where at most two excitations move from one register to another, in transitions such as
\begin{align}\label{eq:4-local}
\ket{\mathord{\bullet}~\mathord{\bullet}~\mathord{\circ}~\mathord{\circ}} \leftrightarrow \ket{\mathord{\circ}~\mathord{\circ}~\mathord{\bullet}~\mathord{\bullet}}\,.
\end{align}
We also require transitions where two excitations are combined into one or vice versa.
    \begin{align}\label{eq:3-local}
\ket{\bullet\bullet\circ} \leftrightarrow \ket{\circ\circ\bullet}\,.
\end{align}
In the more practical discussion of \cref{sec:experiment}, we discuss how here it is important for the outgoing excitation's inherent energy (e.g., the frequency of the transmon that realizes it) to be exactly twice the energy of two incoming excitations, so that energy is conserved. 

Each of these transitions corresponds to one of the perfect state transfer terms in \cref{eq:PST}, but its realization requires 4-local or 3-local interactions for \cref{eq:4-local} and \cref{eq:3-local}, respectively. This is why our construction has locality 4.

A basic implementation does not also achieve degree 4, because the address information deposited as an excitation at a particular intermediate node of the tree in \cref{fig:bucket-brigade} must then interact with many subsequent excitations descending to lower levels of the tree. Our solution to this is to introduce what we call the ``switch wire'' register at each binary tree node. This is a collection of qubits arranged in 1D---one 1D rail for $\ket{0}$ and one 1D rail for $\ket{1}$.   The deposited address information begins in the first slot or ``slice'' of this register. The first time this excitation interacts with another excitation, it advances to the second slot in the switch wire, and so forth. This prevents any individual slice from needing to interact with more than 4 other registers. This same trick is used in other aspects of our construction, as well, ensuring each configuration encountered along the path is unique and easily determined from the last by local logic, enforceable through perfect state transfer terms.

\subsection{Noise resilience}
Now we address the issue of passively implementing the time-independent Hamiltonian.
The Hamiltonian acts on $\Theta(N)$ auxiliary qubits, so at first it seems that
one must prepare and protect all of them for every query. Fortunately, this is not true. The key point is that when we make a query on a given basis state $\ket{x}$ in the address register, the query is sensitive to  errors
only in a small part of the device. 

\paragraph{Branch isolation.}
Fix an address $x$.  During the ideal evolution, only
$F_n=\Theta(n^2)=\Theta(\log^2N)$ physical qubits can affect the root-to-leaf path for this address.
We call this set its \emph{footprint}. 
We show that in a formal sense, errors in disjoint pairs of footprints do not affect each other.
This is the continuous-time version of branch isolation for bucket-brigade QRAM~\cite{HannNoise21}.

Branch isolation is then used for all our noise bounds.
We never sum up the errors
over the full $\Theta(N)$ hardware, only the errors that affect the
branch being queried.

\paragraph{Thermal initialization and reuse.}
The physical Hamiltonian is $K=H+P$, where $P$ contains fixed on-site energy
penalties chosen so that $[H,P]=0$. We will use this Hamiltonian for our analysis of noise resilience. Let $\Delta$ be the gap between the vacant and excited states of a qubit, and let
$\kappa$ denote how much couplings can lower this energy.  The
effective gap is $g=\Delta-\kappa$.  At inverse temperature $\beta$, a fixed
footprint contains a defect with probability at most $F_ne^{-\beta g}$.

On the complementary event, the query is exact and returns the retained device
to the same state.
One can show that over $m$ queries,
\[
    \epsilon_m\le4m\sqrt{F_ne^{-\beta g}}.
\]
Hence
\[
    \beta g\ge\log F_n+2\log(4m/\varepsilon)
\]
suffices for error $\varepsilon$.  If $m,\varepsilon^{-1}=\poly(n)$ and
$g$ is constant, then $T=O(1/\log\log N)$. Hence we have to pay the cooling cost only once, and the temperature required scales inversely in $\log m$ for $m$ queries.  Each query uses only $O(n)$ root
operations, giving $O(\log N)$ energy per query. See \cref{sec:thermal} for details of the proof, and \cref{cor:boundary-energy} for a formal claim of the energy cost.

\paragraph{Persistent errors.}
To show that our construction is scalable, i.e. components only need to be implemented with $1/\polylog N$ precision, we bound the error due to imprecision in the Hamiltonian terms.  To show the bound, we compare the physical
device with the same noisy device that evolves \emph{idle} outside the branch being
queried.

For static imprecision, let $W$ be the largest total norm of perturbation terms
that can touch one branch.
Query $k$ can encounter damage accumulated during the first $k$ query
intervals.  Summing this accumulated damage gives
\[
    \epsilon_m\le2m(m+1)\tau W
      =O(m^2\lambda_{\rm err}\log^4N),
\]
where $\lambda_{\rm err}$ is the largest perturbation weight incident on one
physical qubit.

We can generalize this analysis to Markovian noise, modelled as a Lindbladian~\cite{Lindblad1976,GKS1976}. Let $G_j$ be the largest integrated noise exposure of one
branch during query $j$.  Tracking the amplitude for the noisy idle device to
return to its reference state gives
\[
    \epsilon_m
      \le8m\sum_{j=1}^m(m-j+1)G_j.
\]
If the incident noise rate on each qubit is at most $\gamma$, then
\[
    \epsilon_m=O(m^3\gamma\log^4N).
\]

See \cref{sec:imperfect,sec:dynamical} for detailed proofs. The important point is that every bound depends on the
$\Theta(\log^2N)$ branch footprint rather than on the full memory size.
For $m=\poly(n)$, the component precision is inverse-polynomial in $n=\log N$.

\section{The Hamiltonian}
\label{sec:hamiltonian}

The Hamiltonian has many terms. Each term is associated with an index $j=1,\ldots, L$, where $L$ is the ``schedule length''. The set of all terms associated with step $j$ are denoted $\mathcal{T}_j$, and all terms in $\mathcal{T}_j$ receive the same $j$-dependent weight designed to enforce perfect state transfer.  We write:
\begin{equation}
    H=\sum_{j=1}^{L} J_j\sum_{T\in\mathcal T_j}\left(T+T^\dagger\right).
\end{equation}
Where $J_j=\frac{\Omega}{2}\sqrt{j(L-j+1)}$, for $\Omega = \Theta(1/L)$. This equation is the decomposition used for the locality and degree bounds.  Each summand $J_j(T+T^\dagger)$ is Hermitian. The terms $T$ are described explicitly in \cref{sec:hamiltonian_terms}, but we first set up the notation in \cref{sec:notation}. There will be three kinds of terms: Exchange, Route, and Read. In addition to the address, bus, and memory qubits, these terms will act on the carrier and switch qubits.

\subsection{Evolution \& Notation}
\label{sec:notation}
In this section, we introduce the notation we need to specify the construction of the Hamiltonian. Everywhere below, single qubits encode whether a physical excitation is present in their mode as $\ket{\bullet}$ and its absence by $\ket{\circ}$. Registers will be comprised of 2 or 3 such qubits, and the configuration of excitations will record a value denoted by $\ket{\emptyset}, \ket{\star}, \ket{0}$ or $\ket{1}$.
Before giving the formal construction of the Hamiltonian it will be useful to concretize the description of its intended evolution, which follows the bucket-brigade QRAM circuit~\cite{GLL08}. We will be working over a complete binary tree of height $n$, with nodes $v\in\bits^{\le n}$, root $\mathcal{R} = \varnothing$.  The depth of $v$ is $|v|$.

Fix a computational-basis address \(x=x_1\cdots x_n\). The query has three stages. First, the carrier transports the address bits into switches along the branch specified by \(x\): bit \(x_k\) determines the direction taken at depth \(k-1\). Next, the carrier takes the bus value to leaf \(x\), reads the memory bit \(D_x\), and returns the updated value to the bus. Finally, it restores the address bits and empties the switches. This is visualized in \cref{fig:bucket-brigade}. We now describe the registers that implement this sequence. We first define the input and output registers, which are called address and bus.

\paragraph{Address and bus.}
The address registers $A_1,\ldots,A_n$ encode the input to the query. The query will eventually be read out from the bus register $B$. Each of them has three valid states, \begin{equation}
    \ket{\emptyset}=\ket{\circ\circ},\qquad
    \ket{0}=\ket{\bullet\circ},\qquad
    \ket{1}=\ket{\circ\bullet}.
\end{equation}

Now we will define the carrier and switch registers. The number of these depends on the evolution time of the Hamiltonian, so it will be useful to compute that first. The total evolution time is
\begin{equation}
    L = \sum_{k=1}^n 2k + (n+1) + 1 + (n+1) + \sum_{k=1}^n 2k = 2n^2+O(n)
    \label{eq:schedule-length}
\end{equation}
Loading $A_k$ takes $2k$ transitions: make the carrier bit, descend $k-1$ levels, store it, and return.  The query takes $2n+3$ transitions: absorb the bus value into the carrier, traverse to the leaf and read, return, and deposit the result back into $B$.  Unloading reverses the load, taking $2k$ transitions for bit $k$.

For $v$ at depth $d$, it is visited twice for each load, twice during the query, and twice for each unload. Let $R_v$ be the number of times $v$ is visited, then
\begin{equation}
    R_v = 2(n-d-1)+2+2(n-d-1)=4(n-d)-2,
\end{equation}

\paragraph{Carrier registers.}
We will have $O(N)$ carrier registers for many $(j,v)$ (step, node) pairs in the evolution. The carrier registers have 4 possible values. They can either be ``inactive" ($\ket \emptyset$) or ``active" and take the values $\ket{\star}, \ket{0}, \ket{1}$. Here we may imagine a register is active if it is occupied by a fermion or boson. Throughout the evolution we will maintain the invariant that there is a single active carrier register. This will be interpreted as the \emph{carrier} moving through the different nodes.

We encode each small register in unary. Every carrier register is realized by three qubits.  For every $j$ and every node $v$ in the binary tree such that $v$ could contain the carrier at step $j$, define the carrier register
\begin{equation}
    \car_{j,v}=(q^\star_{j,v},q^0_{j,v},q^1_{j,v}).
\end{equation}
The valid states of one register are
\begin{equation}
    \ket{\emptyset}=\ket{\circ\circ\circ},\qquad
    \ket\star=\ket{\bullet\circ\circ},\qquad
    \ket{0}=\ket{\circ\bullet\circ},\qquad
    \ket{1}=\ket{\circ\circ\bullet}.
\end{equation}
A valid history state has exactly one active carrier register across all layers and sites.  Initially this carrier is $\ket\star$ at the root. The active carrier register plays a dual role. Its position acts as a clock, tracking the step of the evolution. And its state encodes for example whether there is an address bit currently being routed, and if so what is its value.

\paragraph{Switch wires.}
For each internal node $v\in\bits^{<n}$ recall that $R_v=4(n-|v|) - 2$. This number is how many times the carrier will visit this node on a valid evolution of the Hamiltonian if $v$ lies on the branch of the tree corresponding to the query. Define a sequence of \emph{switch wire} registers

\begin{equation}
    Y_{v,0},Y_{v,1},\ldots,Y_{v,R_v},
    \qquad
    Y_{v,r}=(q^0_{v,r},q^1_{v,r}),
\end{equation}
with valid states
\begin{equation}
    \ket{\emptyset}=\ket{\circ\circ},\qquad
    \ket{0}=\ket{\bullet\circ},\qquad
    \ket{1}=\ket{\circ\bullet}.
\end{equation}
All switch registers are initialized to $\ket{\emptyset}$. The value \(c\in\{0,1\}\) is encoded by an excitation in qubit \(q^c\). When node $v$ stores the address bit $a_{|v|+1}$, the excitation starts in qubit $Y_{v,0}$ and advances from $Y_{v,i-1}$ to $Y_{v, i}$ each time $v$ is acted on by a Hamiltonian term. Unrolling the counter of how many times this node has been visited into many registers encoded in unary is what keeps the interaction hypergraph degree 4.

\paragraph{Schedule DAG.}
The schedule records the carrier's position after each elementary
operation. At each step, the carrier either stays at the same tree node,
moves to a child, or moves to its parent. The order of these operations
and their depths are fixed independently of the input address.

We represent the schedule by a layered directed acyclic graph (DAG) with
layers
\[
    U_j=\{(j,v):v\text{ is a possible carrier position after step }j\},
    \qquad j=0,\ldots,L.
\]
Vertex $(j,v)$ corresponds to the carrier register $\car_{j,v}$.
Every edge goes from layer $j-1$ to layer $j$. Depending on the
scheduled operation, the edges have one of three forms:
\[
\begin{aligned}
    (j-1,v)&\longrightarrow(j,v)
        &&\text{stay at the same tree node},\\
    (j-1,v)&\longrightarrow(j,vc)
        &&\text{move to child }vc, ~c\in\bits,\\
    (j-1,vc)&\longrightarrow(j,v)
        &&\text{move to the parent}.
\end{aligned}
\]
At each step we include all edges of the prescribed form at the
prescribed depth. Even an operation that stays at the same tree node
moves between distinct carrier registers, since the layer changes.
The local terms defined below specify how each transition changes
the values held by the registers.

For a basis address $x$, the schedule and the directions specified
by $x$ determine a branch walk
\[
    u_0(x),\ldots,u_L(x),
    \qquad u_j(x)=(j,v_j(x)),
\]
starting at $(0,\mathcal R)$.
When we define the terms in \cref{sec:hamiltonian_terms}, $\car_s$ and $\car_t$ 
will denote the source and target carrier registers of the edge under discussion.

\paragraph{Memory.}
For each leaf $x\in\bits^n$ there is a memory qubit $M_x$, initialized as $\ket{D_x}$.  The Hamiltonian acts diagonally on memory and never changes $D$. Since this qubit only needs to store classical data, we can imagine these registers as classical memory and will not worry about their initialization or coherence.

\paragraph{Elementary operators.}
We will define the operators which will help us define the Hamiltonian in \cref{sec:hamiltonian_terms}, and we will use the fact that we encoded our states in unary. Each register is a block of qubits (a carrier register is the three qubits
$q^\star,q^0,q^1$; a switch slice is the two qubits $q^0,q^1$), and every term of the
Hamiltonian is a product of single-qubit ladder operators, one per qubit it touches. On a
single qubit, $\sgp=\lvert\bullet\rangle\!\langle\circ\rvert$ raises and $\sgm=\lvert\circ\rangle\!\langle\bullet\rvert$
lowers. In other words, $\sgp$ creates an excitation and $\sgm$ removes one, with $\sgp\ket\bullet=0$ and
$\sgm\ket\circ=0$. We address an individual qubit by its register and its value: for a register
$R$ and a value $c$ in $R$'s alphabet, write
\begin{equation}
  \cre{R}{c}\ \text{and}\ \ann{R}{c}
  \qquad\text{for }\sgp\ \text{and}\ \sgm\ \text{on the qubit }q^c\text{ of }R,
  \label{eq:op-dictionary}
\end{equation}
so $\cre{R}{c}$ writes value $c$ into $R$ and $\ann{R}{c}$ removes it. Relabeling
a register from value $c'$ to $c$ is then the product $\cre{R}{c}\ann{R}{c'}$, since the unary
encoding stores the value in which qubit is up. A term fires only when its lowering operators meet
occupied (active) qubits.

\subsection{Construction}
\label{sec:hamiltonian_terms}

Now we recall the weights needed to construct a perfect state transfer Hamiltonian with evolution time $L$,
\begin{equation}
    J_j=\frac{\Omega}{2}\sqrt{j(L-j+1)},
    \qquad
    \Omega=\Theta(1/L).
    \label{eq:couplings}
\end{equation}
For each layer $j$ there is a set $\mathcal T_j$ of terms at step $j$ if the schedule.
Recall that the Hamiltonian is
\begin{equation}
    H=\sum_{j=1}^{L} J_j\sum_{T\in\mathcal T_j}\left(T+T^\dagger\right).
    \label{eq:H}
\end{equation}
The source and target carrier registers are those of the corresponding edge in the layered graph we defined in the previous subsection. They are denoted by $\car_s$ and $\car_t$ inside the local formulas.

There are only three kinds of terms: exchange a bit with the carrier,
route through a switch, and read a memory bit.

\paragraph{Exchange a bit with the carrier.}
Let $R$ be any register with states
$\ket\emptyset,\ket0,\ket1$.  An exchange layer moves one bit between $R$
and the carrier:
\begin{equation}
\begin{aligned}
    \mathcal T_j
      &=\left\{
        \cre{\car_t}{c}\ann{\car_s}{\star}\ann{R}{c}:c\in\bits
      \right\},
      &&\text{if the bit moves from $R$ to the carrier},\\
    \mathcal T_j
      &=\left\{
        \cre{\car_t}{\star}\ann{\car_s}{c}\cre{R}{c}:c\in\bits
      \right\},
      &&\text{if the bit moves from the carrier to $R$}.
\end{aligned}
\label{eq:carrier-exchange}
\end{equation}
On valid states these are the two directions of
\[
    \ket\star_{\car_s}\ket c_R
      \longleftrightarrow
    \ket c_{\car_t}\ket\emptyset_R.
\]

We use this same interaction on three kinds of registers.  At the root, an
address bit is transferred from $A_k$ into the carrier before it is routed,
and returned to $A_k$ after that part of the schedule is reversed.  At a tree
node, the carrier transfers the address bit into the first slice $Y_{v,0}$ of
the switch wire; after the later routing steps have advanced that excitation
to $Y_{v,R_v}$, the carrier takes it back.  Finally, the bus value is
transferred into the carrier before the leaf traversal and returned to $B$
after the carrier comes back to the root.  Thus all three uses are the same
three-local interaction.

\paragraph{Route through a switch.}
At a routing layer $j$, let $S_j\subseteq\{\star,0,1\}$ be the carrier values
that can occur on a valid history.  Thus $S_j=\{\star\}$ when the carrier
holds no bit and $S_j=\{0,1\}$ when it carries an address or bus bit.  In
particular, $|S_j|\le2$.

Suppose the layer routes through depth $d$.  For $v\in\bits^d$, let
$r=r_j(v)$ be the active switch slice.  For each child $b\in\bits$ and
$s\in S_j$, include
\begin{equation}
    R_{v,r,b,s}
      =\cre{\car_t}{s}\ann{\car_s}{s}
       \ann{Y_{v,r}}{b}\cre{Y_{v,r+1}}{b},
    \qquad
    \mathcal T_j
      =\{R_{v,r_j(v),b,s}:v\in\bits^d,\ b\in\bits,\ s\in S_j\}.
    \label{eq:route}
\end{equation}
For a downward route the carrier edge is $v\to vb$; for an upward route it is
$vb\to v$.  The term moves the carrier along the child selected by the stored
bit and advances that switch excitation by one slice.  The wrong-child term
annihilates the state.  This interaction is four-local: it touches the source
and target carrier qubits and the old and new switch qubits.

\paragraph{Leaf read.}
For each leaf $x$ and $c\in \bits$, include
\begin{equation}
    Q_{x,c}
      =\cre{\car_t}{c}\ann{\car_s}{c}\otimes\proj{0}_{M_x}
       +\cre{\car_t}{c\xor1}\ann{\car_s}{c}\otimes\proj{1}_{M_x},
    \qquad
    \mathcal T_j=\{Q_{x,c}:x\in\bits^n,\ c\in C_j\}.
    \label{eq:leaf-read}
\end{equation}
On the memory basis state $\ket{D_x}$,
\begin{equation}
    \ket c_{\car_s}\ket{D_x}_{M_x}
      \longmapsto
    \ket{c\xor D_x}_{\car_t}\ket{D_x}_{M_x}.
\end{equation}
Every leaf has the same controlled flip, so the Hamiltonian is independent of
$D$.  The term is four-local: it touches one source carrier qubit, both target
bit qubits, and $M_x$.

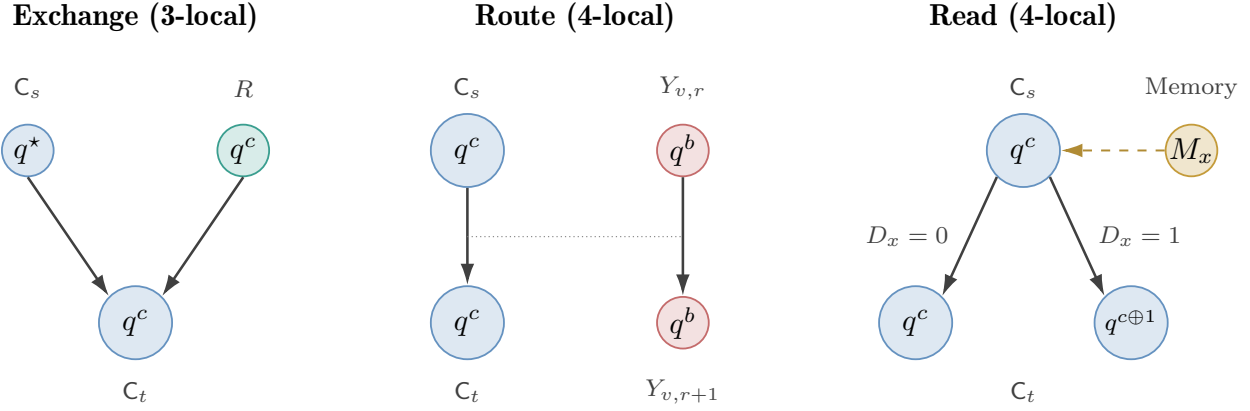
\begin{figure}[H]
\centering
\resizebox{\textwidth}{!}{%
\begin{tikzpicture}[
  >={Latex[length=2.5mm,width=1.6mm]},
  light/.style={
    circle,draw,line width=.6pt,
    minimum size=6mm,inner sep=0pt,font=\small
  },
  heavy/.style={
    light,minimum size=8.485281mm,
    fill=clkblue!18,draw=clkblue!85
  },
  carrier/.style={
    light,fill=clkblue!18,draw=clkblue!85
  },
  flow/.style={
    ->,line width=.85pt,draw=black!75
  },
  control/.style={
    ->,dashed,line width=.65pt,draw=memgold!90!black
  },
  label/.style={
    font=\scriptsize,text=black!75
  },
  title/.style={
    font=\small\fontseries{b}\selectfont
  },
  note/.style={
    font=\scriptsize,align=center,text=black!80
  }
]

\begin{scope}
  \node[title] at (1.25,2.55) {Exchange (3-local)};

  \node[label] at (0,1.72) {$\car_s$};
  \node[label] at (2.5,1.72) {$R$};
  \node[label] at (1.25,-1.8) {$\car_t$};

  \node[carrier] (es) at (0,1) {$q^\star$};
  \node[light,fill=posteal!20,draw=posteal]
    (er) at (2.5,1) {$q^c$};
  \node[heavy] (et) at (1.25,-1) {$q^c$};

  \draw[flow] (es.south) -- (et.north west);
  \draw[flow] (er.south) -- (et.north east);
\end{scope}

\begin{scope}[shift={(5.1,0)}]
  \node[title] at (1.25,2.55) {Route (4-local)};

  \node[label] at (0,1.72) {$\car_s$};
  \node[label] at (2.5,1.72) {$Y_{v,r}$};
  \node[label] at (0,-1.8) {$\car_t$};
  \node[label] at (2.5,-1.8) {$Y_{v,r+1}$};

  \node[heavy] (rs) at (0,1) {$q^c$};
  \node[heavy] (rt) at (0,-1) {$q^c$};
  \node[light,fill=swrose!20,draw=swrose]
    (ys) at (2.5,1) {$q^b$};
  \node[light,fill=swrose!20,draw=swrose]
    (yt) at (2.5,-1) {$q^b$};

  \draw[flow] (rs.south) -- (rt.north);
  \draw[flow] (ys.south) -- (yt.north);
  \draw[densely dotted,black!45] (0,0) -- (2.5,0);

\end{scope}

\begin{scope}[shift={(10.3,0)}]
  \node[title] at (1.25,2.55) {Read (4-local)};

  \node[label] at (1.25,1.72) {$\car_s$};
  \node[label] at (3.2,1.72) {Memory};
  \node[label] at (1.25,-1.8) {$\car_t$};

  \node[heavy] (qs) at (1.25,1) {$q^c$};
  \node[light,fill=memgold!25,draw=memgold]
    (mem) at (3.2,1) {$M_x$};
  \node[heavy] (q0) at (0,-1) {$q^c$};
  \node[heavy,font=\scriptsize]
    (q1) at (2.5,-1) {$q^{c\oplus1}$};

  \draw[flow] (qs.south west) -- (q0.north east);
  \draw[flow] (qs.south east) -- (q1.north west);
  \draw[control] (mem.west) -- (qs.east);

  \node[label] at (-.1,0) {$D_x=0$};
  \node[label] at (2.6,0) {$D_x=1$};
\end{scope}

\end{tikzpicture}%
}
\caption{The three kinds of Hamiltonian terms.}
\label{fig:locality}
\end{figure}

\subsection{Correctness}

Fix a computational basis query $x\in\bits^n$, initial bus value $b\in\bits$, and memory string $D\in\bits^N$.  Let $\ket{\Phi_j(x,b,D)}$ be the intended configuration after $j$ scheduled steps.
The address and bus values occupy the registers prescribed by the
schedule, while the memory remains $\ket D_M$. There is one active
carrier register on layer $j$, and the active switch slices are
those prescribed by the schedule.

We will now state the properties of the valid states of the system during ideal evolution. To do so, we define the initial configuration $\ket{\Phi_0(x,b,D)}$ explicitly.
It has address $\ket x_A$, bus $\ket b_B$, memory $\ket D_M$, the root carrier
in state $\ket\star$, and every other carrier and switch register in the state $\ket{\emptyset}$.

\begin{lemma}[Valid states]
\label{lem:validity-induction}
For every $j=0,1,\ldots,L$, the intended state
$\ket{\Phi_j(x,b,D)}$ satisfies:
\begin{enumerate}[label=(\alph*),leftmargin=2.2em,itemsep=0.15em]
    \item Exactly one carrier register is active, namely the layer-$j$ carrier.
    \item Every other carrier register is vacant.
    \item Each switch wire is either vacant or contains one excitation carrying
    $x_{|v|+1}$ in the slice prescribed by the schedule.
    \item The address and bus registers contain the values prescribed by the
    schedule, or are vacant while their excitation is carried through the tree.
    The memory remains $\ket D_M$.
    \item The target carrier register, and the target switch slice for a routing
    step, is vacant.
\end{enumerate}
Moreover, exactly one forward term in $\mathcal T_{j+1}$ acts nontrivially on
$\ket{\Phi_j}$, and it maps this state to $\ket{\Phi_{j+1}}$.
\end{lemma}

\begin{proof}
We induct on $j$.  The claim holds for $\Phi_0$ by definition.  Assume it
holds at layer $j$.  There are three possible kinds of next step.

\paragraph{Exchange.}
By \eqref{eq:carrier-exchange}, the exchange acts on a two-rail register $R$.
Exactly one rail of $R$ or the carrier contains the bit $c$ prescribed by the
schedule.  Hence only the term with that value of $c$ acts; the other term
lowers an empty qubit.  The exchange moves this bit between $R$ and the
carrier and leaves all other registers unchanged.  This covers exchanges
with the address, a switch wire, and the bus.

\paragraph{Route.}
Suppose the carrier passes through node $v$, and put
$\beta=x_{|v|+1}$.  The switch excitation occupies the $\beta$ rail of the
current slice.  A term for the other child lowers the empty
$(1-\beta)$ rail, while a term for the wrong carrier value lowers an empty
carrier qubit.  Thus exactly one routing term acts:
\[
    \ket{s}_{\car_s}\ket{\emptyset}_{\car_t}
    \ket{\beta}_{Y_{v,r}}\ket{\emptyset}_{Y_{v,r+1}}
    \longmapsto
    \ket{\emptyset}_{\car_s}\ket{s}_{\car_t}
    \ket{\emptyset}_{Y_{v,r}}\ket{\beta}_{Y_{v,r+1}}.
\]
The schedule uses the switch slices in increasing order, so
$Y_{v,r+1}$ is vacant.

\paragraph{Leaf read.}
At the leaf, only the term whose source carrier qubit is occupied can act.
The memory is still $\ket{D_x}$, so this term maps
\[
    \ket c_{\car_s}\ket{\emptyset}_{\car_t}\ket{D_x}_{M_x}
    \longmapsto
    \ket{\emptyset}_{\car_s}
    \ket{c\xor D_x}_{\car_t}\ket{D_x}_{M_x}.
\]
Thus the carrier returns carrying $b\xor D_x$.

In every case, the unique surviving term performs exactly the next scheduled
transition, vacates its source, and leaves the next target vacant.  Hence
$\ket{\Phi_{j+1}}$ satisfies the same invariants.  The claim follows by
induction.
\end{proof}

\begin{lemma}[Unique predecessor]
\label{lem:unique-predecessor}
For every $j=1,\ldots,L$ and every intended branch state $\ket{\Phi_j(x,b,D)}$, among the adjoints $T^\dagger$ with $T\in\mathcal T_j$ exactly one acts nontrivially.  The forward term maps $\ket{\Phi_{j-1}(x,b,D)}$ to $\ket{\Phi_j(x,b,D)}$, and its adjoint maps $\ket{\Phi_j(x,b,D)}$ back to $\ket{\Phi_{j-1}(x,b,D)}$.
\end{lemma}

\begin{proof}
Every transition in \cref{sec:hamiltonian_terms} moves the active carrier from
layer $j-1$ to layer $j$.  Hence $T^\dagger$ can act on $\Phi_j$ only if it
lowers the unique occupied layer-$j$ carrier qubit.  Among those candidates,
the valid state fixes everything else: an exchange term fixes the exchanged bit, a
routing term fixes the carrier value and occupied switch qubit, and a leaf-read term is
reversible once $D_x$ is fixed.  Thus only the reverse of the transition that
created $\Phi_j$ survives.
\end{proof}

\begin{proposition}[Invariant transfer chains]\label{prop:invariant}
For every $x,b,D$, the subspace
\begin{equation}
    \mathcal K_{x,b,D}=\operatorname{span}\{\ket{\Phi_0(x,b,D)},\ldots,\ket{\Phi_L(x,b,D)}\}
\end{equation}
is invariant under $H$, and on this subspace
\begin{equation}
    H\big|_{\mathcal K_{x,b,D}}
    =\sum_{j=1}^{L}J_j
    \left(
       \ket{\Phi_j}\!\bra{\Phi_{j-1}}
       +\ket{\Phi_{j-1}}\!\bra{\Phi_j}
    \right).
    \label{eq:path-restriction}
\end{equation}
\end{proposition}

\begin{proof}
The states $\ket{\Phi_j}$ are mutually orthogonal because at each step $j$, a different carrier register is active.  Consider $H\ket{\Phi_m}$.  A summand from layer $j$ can act nontrivially only if its source carrier register is the active layer-$m$ register or if the source carrier register of its adjoint is the active layer-$m$ register.  Since all carrier registers are layer-indexed, this requires $j=m+1$ for a forward action or $j=m$ for a backward action.

By \cref{lem:validity-induction}, among the forward terms in $\mathcal T_{m+1}$ exactly one survives and maps $\ket{\Phi_m}$ to $\ket{\Phi_{m+1}}$.  By \cref{lem:unique-predecessor}, among the adjoints of terms in $\mathcal T_m$ exactly one survives and maps $\ket{\Phi_m}$ to $\ket{\Phi_{m-1}}$.  Therefore, for $0<m<L$,
\begin{equation}
    H\ket{\Phi_m}=J_{m+1}\ket{\Phi_{m+1}}+J_m\ket{\Phi_{m-1}},
\end{equation}
with the evident boundary cases at $m=0$ and $m=L$.  This is exactly~\eqref{eq:path-restriction}.
\end{proof}

\begin{lemma}[Deterministic transfer on the history path]
\label{lem:pst}
Let
\begin{equation}
    H_{\rm path}=\sum_{j=1}^{L}\frac{\Omega}{2}\sqrt{j(L-j+1)}
    \left(\ket j\!\bra{j-1}+\ket{j-1}\!\bra j\right).
\end{equation}
Then
\begin{equation}
    e^{-iH_{\rm path}\pi/\Omega}\ket0=e^{i\varphi_L}\ket L
    \label{eq:deterministic_state_transfer_eigenvalue_eq}
\end{equation}
for a phase $\varphi_L$ depending only on $L$.
\end{lemma}

This is the perfect state transfer of Christandl, Datta, Ekert, and
Landahl~\cite{Christandl2004}.
We essentially reproduce their proof in modern language below.

\begin{proof}
The point is that the weighted path is implicitly encoding $L$ non-interacting
qubits, seen only through their total Hamming weight. So we can interpret the evolution as $L$ independent bit flips happening in lockstep.

Work in $(\bits)^{\otimes L}$, and for a string $s\in\bits^L$ write $\lvert
s\rvert$ for its Hamming weight. Group the $2^L$ basis strings by weight into the
uniform superpositions (aka Dicke states)
\begin{equation}
    \ket{D_j}=\binom{L}{j}^{-1/2}\!\!\sum_{\lvert s\rvert=j}\ket s,
    \qquad j=0,1,\dots,L,
\end{equation}
and identify the path site $\ket j$ with $\ket{D_j}$; the two ends are the product
strings $\ket0=\ket{0^L}$ and $\ket L=\ket{1^L}$. Let $X_i$ be the bit flip on
coordinate $i$ and set $B=\tfrac{\Omega}{2}\sum_{i=1}^L X_i$. We claim $B$ agrees
with $H_{\rm path}$ on the span of the $\ket{D_j}$, and then compute $e^{-iB\tau}$
directly.

\paragraph{$B$ is $H_{\rm path}$ on the Dicke state subspace.}
Each $X_i$ flips one coordinate, changing the weight by exactly $\pm1$, so $B$
carries weight-$j$ strings to weight $j\pm1$: on the basis $\{\ket{D_j}\}$ it is
tridiagonal with zero diagonal, like $H_{\rm path}$. Its upper off-diagonal entry
is an edge count. We think of the sensitivity graph on this Boolean hypercube (two strings have an edge when they differ by one bit). Between the layers corresponding to Hamming weight $j$ and $j-1$ we have exactly $\binom{L}{j}\,j$ edges,
since each of the $\binom{L}{j}$ weight-$j$ strings has $j$ ones to unset.
As every term of $B$ contributes weight $\tfrac{\Omega}{2}$ per edge and the $\ket{D_j}$ are
flat unit vectors,
\begin{equation}
    \bra{D_j}B\ket{D_{j-1}}
    =\frac{\Omega}{2}\cdot\frac{\binom{L}{j}\,j}{\sqrt{\binom{L}{j}\binom{L}{j-1}}}
    =\frac{\Omega}{2}\sqrt{j(L-j+1)},
\end{equation}
using $\binom{L}{j}/\binom{L}{j-1}=(L-j+1)/j$. These are the couplings of
$H_{\rm path}$, proving the claim. Since $\ket0=\ket{D_0}$ lies in the weight
subspace and $B$ preserves it, we may evolve with $B$ in place of $H_{\rm path}$.

\paragraph{Evolving the product state.}
Now the implicit encoding becomes convenient. $\ket0=\ket{0^L}$ is a product state and the $X_i$
commute, so the evolution factorizes into $L$ independent single-qubit rotations,
\begin{equation}
    e^{-iB\tau}\ket{0^L}
    =\bigotimes_{i=1}^{L}e^{-i(\Omega\tau/2)X}\ket0 .
\end{equation}
Each qubit experiences a roation as $e^{-i\theta X}\ket0=\cos\theta\,\ket0
-i\sin\theta\,\ket1$. At the transfer time $\tau=\pi/\Omega$ the angle is
$\theta=\pi/2$, so every qubit flips completely, $e^{-i(\pi/2)X}\ket0=-i\ket1$.
Therefore
\begin{equation}
    e^{-iH_{\rm path}\pi/\Omega}\ket0
    =e^{-iB\pi/\Omega}\ket{0^L}
    =(-i)^{L}\ket{1^L}
    =e^{i\varphi_L}\ket L,
    \qquad \varphi_L=-\tfrac{\pi L}{2}.
\end{equation}
The phase $\varphi_L$ depends only on $L$, as claimed.
\end{proof}

\begin{figure}[t]
\centering
\begin{tikzpicture}[
  font=\footnotesize,
  hub/.style={draw, circle, minimum size=9mm, inner sep=0pt, thick, font=\small},
  tag/.style={draw=black!20!white, fill=Cerulean!3!white, thick,
             rounded corners=3pt, inner sep=3pt, font=\scriptsize, align=center},
  spoke/.style={line width=0.65pt, black!45},
  title/.style={font=\footnotesize\bfseries},
  deg/.style={font=\scriptsize, text=black!60}
]

\begin{scope}[shift={(0,0)}]
  \node[title] at (0,2.05) {Carrier qubit};
  \node[deg] at (0,1.65) {degree $\le 4$};
  \node[hub, draw=clkblue, fill=clkblue!20] (c) at (0,0) {$q^c$};
  \foreach \ang/\txt in {
      135/{routing in}, 45/{routing in},
      225/{routing out}, 315/{routing out}} {
    \node[tag] (t) at (\ang:1.55) {\txt};
    \draw[spoke] (c)--(t);
  }
\end{scope}

\begin{scope}[shift={(5.2,0)}]
  \node[title] at (0,2.05) {Interior switch qubit};
  \node[deg] at (0,1.65) {degree $\le 4$};
  \node[hub, draw=swrose, fill=swrose!20] (s) at (0,0) {$q^b$};
  \foreach \ang/\txt in {
      135/{routing in\\$s=0$}, 45/{routing in\\$s=1$},
      225/{routing out\\$s=0$}, 315/{routing out\\$s=1$}} {
    \node[tag] (t) at (\ang:1.55) {\txt};
    \draw[spoke] (s)--(t);
  }
\end{scope}

\begin{scope}[shift={(10.4,0)}]
  \node[title] at (0,2.05) {End switch qubit};
  \node[deg] at (0,1.65) {degree $\le 3$};
  \node[hub, draw=swrose, fill=swrose!20] (e) at (0,0) {$q^b$};
  \foreach \ang/\txt in {
      90/{exchange}, 210/{routing\\$s=0$}, 330/{routing\\$s=1$}} {
    \node[tag] (t) at (\ang:1.55) {\txt};
    \draw[spoke] (e)--(t);
  }
\end{scope}

\begin{scope}[shift={(2.6,-4.1)}]
  \node[title] at (0,1.85) {Address or bus qubit};
  \node[deg] at (0,1.45) {degree $2$};
  \node[hub, draw=posteal, fill=posteal!20] (r) at (0,0) {$q^c_R$};
  \node[tag] (a) at (150:1.45) {exchange\\to carrier};
  \node[tag] (b) at (30:1.45) {exchange\\from carrier};
  \draw[spoke] (r)--(a);
  \draw[spoke] (r)--(b);
\end{scope}

\begin{scope}[shift={(7.8,-4.1)}]
  \node[title] at (0,1.85) {Memory qubit};
  \node[deg] at (0,1.45) {degree $2$};
  \node[hub, draw=memgold, fill=memgold!25] (m) at (0,0) {$M_x$};
  \node[tag] (a) at (150:1.45) {leaf read\\$c=0$};
  \node[tag] (b) at (30:1.45) {leaf read\\$c=1$};
  \draw[spoke] (m)--(a);
  \draw[spoke] (m)--(b);
\end{scope}

\end{tikzpicture}
\caption{Number of Hamiltonian terms incident on each qubit type.  A carrier or
interior switch qubit can participate in four routing terms.  An end switch
qubit participates in one exchange term and at most two routing terms.  Each
address or bus qubit participates in two exchange terms, and each memory qubit
in two leaf-read terms.}
\label{fig:degree}
\end{figure}
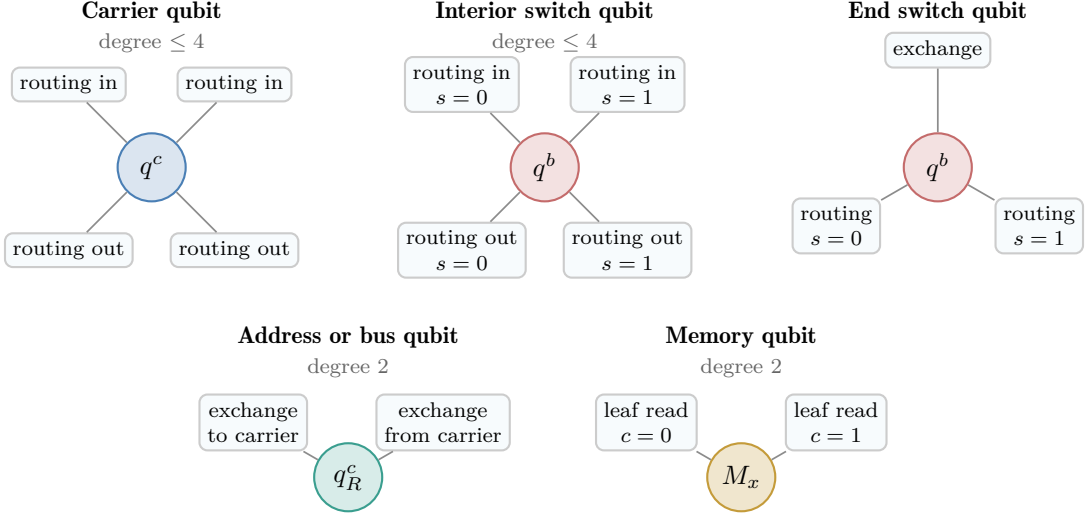

\begin{proof}[Proof of \cref{thm:exact qram}]
Fix a branch $(x,b,D)$ and define the isometry
  \begin{equation}
      V:\mathbb C^{L+1}\to\mathcal K_{x,b,D},
      \qquad
      V\ket j=\ket{\Phi_j(x,b,D)},\quad j=0,\dots,L .
  \end{equation}
By \cref{prop:invariant}, under the identification
$V\ket j=\ket{\Phi_j(x,b,D)}$, the restriction of $H$ to this
history path is the Hamiltonian $H_{\rm path}$ of \cref{lem:pst}:
\begin{equation}
  V\,H_{\rm path}\,V^\dagger=H\big|_{\mathcal K_{x,b,D}} .
  \label{eq:intertwine}
\end{equation}
Because $H$ is Hermitian and $\mathcal K_{x,b,D}$ is $H$-invariant (again
  \cref{prop:invariant}),
  the orthogonal complement is invariant too, so $H$ is block diagonal with respect to
  $\mathcal K_{x,b,D}\oplus\mathcal K_{x,b,D}^{\perp}$. Combining this with
  \eqref{eq:intertwine} and applying \cref{lem:pst} at $\tau=\pi/\Omega$,
  \begin{equation}
      e^{-iH\pi/\Omega}\ket{\Phi_0(x,b,D)}
      =V\,e^{-iH_{\rm path}\pi/\Omega}\ket0
      =V\bigl(e^{i\varphi_L}\ket L\bigr)
      =e^{i\varphi_L}\ket{\Phi_L(x,b,D)}.
  \end{equation}
At layer $L$, all switch wires are vacant, the carrier is $\ket\star$ at the output root site, the address remains $x$, the memory remains $D$, and the bus is $b\xor D_x$.  The path length and couplings are branch-independent, so the phase $\varphi_L$ is independent of $x,b,D$.  Linearity gives~\eqref{eq:xor-qram} for arbitrary superpositions of address and bus basis states.

\paragraph{Locality.} 
By \cref{sec:hamiltonian_terms}, there are only three kinds of Hamiltonian terms.
An exchange term touches two carrier qubits and one qubit of the exchanged register,
so it is three-local.  A routing term touches the source and target carrier qubits
and one qubit in each of two consecutive switch slices, so it is four-local.
A leaf-read term touches one source carrier qubit, both target bit qubits, and one
memory qubit, so it is four-local.  Hence $H$ is four-local, as visualized in
\cref{fig:locality}.

\paragraph{Degree.} We now bound how many terms of~\eqref{eq:H} touch any single
qubit. This is visualized in \cref{fig:degree}, and we justify it here.
A carrier register belongs to one layer, so it is touched only by the move that enters
that layer and the move that leaves it. At a routing layer it can be
in at most two incoming and two outgoing terms, one for each child direction,
so its degree is at most four.  An interior switch qubit can be used by at
most two carrier values when the excitation enters its slice and at most two
when it leaves, again giving degree at most four.  At the first or last switch
slice, one of these events is an exchange, so the degree is at most three.

Each qubit of an address register or the bus participates in exactly two
exchange layers, one in each direction, so its degree is two.  A memory
qubit appears in the two leaf-read terms corresponding to carrier values
$c=0,1$, so its degree is two.
\paragraph{Count.}  Both the carrier registers and switch wire registers are laid out one per (layer, site), so both counts are of the form $\sum_d 2^d\cdot(\text{layers at depth }d)$ with a layer count linear
in $n-d$. Therefore, their count is a multiple of the following sum,
\begin{equation}
    T(n):=\sum_{d=0}^{n-1}2^d(n-d)=2^{n+1}-n-2,
    \label{eq:master-sum}
\end{equation}
where the closed form is the standard evaluation
$\sum_{d=0}^{n-1}2^d(n-d)=n(2^n-1)-\bigl((n-2)2^n+2\bigr)$. Each node $v$ at depth $d$ carries a switch wire of $R_v+1=4(n-d)-1$ slices, and there are
$2^d$ nodes at depth $d$, so the number of switch slices is
\begin{equation}
    \sum_{d=0}^{n-1}2^d\bigl(4(n-d)-1\bigr)
    =4\,T(n)-(2^n-1)
    =7N-4n-7,
    \label{eq:switch-count}
\end{equation}
which is two qubits per slice, which amounts to $14N - O(n)$ switch qubits. A depth-$d$ site hosts the carrier on
$4(n-d)+2$ layers, so, reindexing $d\mapsto d+1$ to match~\eqref{eq:master-sum} and adding the
$4n+4$ layers spent at the root,
\begin{equation}
    (4n+4)+\sum_{d=1}^{n}2^d\bigl(4(n-d)+2\bigr)
    =(4n+4)+8\,T(n)-4(2^n-1)
    =12N-4n-8
    \label{eq:carrier-count}
\end{equation}
carrier registers. At three qubits per register, this is $36N - O(n)$ carrier qubits. Hence the carrier
registers, switch slices, and memory qubits together sum up to $O(N)$.

\paragraph{Evolution time.} Since
$L=\Theta(n^2)$ and $\Omega=\Theta(1/L)$, the transfer time $\tau=\pi/\Omega$ is
$O(n^2)=O(\log^2 N)$.
\end{proof}

\subsection{On-site ``pinning'' term and comments on experimental feasibility}
\label{sec:experiment}

We have described an abstract time-independent Hamiltonian. 
While our formal construction remains hardware agnostic, here we discuss what kinds of interactions have plausible prospects for implementation, and we further specify a few aspects that increase  practicality. 

Conceptually, the key point is that the Hamiltonian terms prescribed by the construction should be fixed into the device at the time of manufacture. We believe this idea is well suited for solid-state platforms like superconducting qubits \cite{Blais2021circuitQED,krantz2019quantum}, where circuit elements are permanently deposited or etched onto the  chip during fabrication. Since the chip has $\Omega(N)$ components, the size of the chip and the associated fabrication cost would scale as $\Omega(N)$, but crucially at most $O(\log N)$ interventions---in the form of input pulses---are necessary during the query itself. Our focus here on solid-state platforms only serves as a way to concretize our general ideas. 

\paragraph{Comments on how the qubits in our construction could be physically realized.} Our construction acts on a collection of localized qubits which can be in the $\ket{\circ}$ or $\ket{\bullet}$ state.  There are many physical systems that offer controllable two-level systems although some may be more suitable for our construction than others. Two examples of physical systems that could be used to realize these qubits are:
\begin{itemize}
    \item Quantum harmonic oscillator: here, the energy levels are equally spaced, which can be interpreted as counting the number of (bosonic) particles in the mode. By truncating to the lowest two energy levels, this could be suitable for implementing our construction, that is, associating $\ket{\circ}$ with the vacuum state (zero particles) and $\lvert\bullet\rangle$ with the first excited state (one particle). All states of higher-occupation are not part of the active Hilbert space---when implemented perfectly, the Hamiltonian would prevent two bosons from  being routed to the same place. For superconducting circuits, the quantum harmonic oscillator Hamiltonian could be realized, for example, by fabricating an LC resonator circuit on the chip, but there are many ways to realize harmonic oscillator Hamiltonians in many quantum computing platforms. 
    \item Transmon \cite{koch2007transmon}: in superconducting circuits, the inductor in the LC circuit may be replaced by a Josephson junction, leading to an anharmonicity in the Hamiltonian. Energy levels are no longer equally spaced, which can enable better addressability of particular levels. We associate $\ket{\circ}$ with the transmon ground state and $\ket{\bullet}$ with one of the excited states, typically the first excited state.
\end{itemize}
Without loss of generality, we may set the ground state energy to zero, and in both cases we may write the two-level Hamiltonian that is realized as simply
\begin{align}
    H_{\rm qubit} = \Delta \ket{\bullet}\!\bra{\bullet} = \Delta n\,, \label{eq:H_qubit}
\end{align}
where $n$ is the particle number operator and $\Delta$ is the energy splitting between the unoccupied and occupied states. 
We may interpret this to mean that the particle which occupies the mode has an inherent energy (or mass), denoted by $\Delta$. The value of $\Delta$ can be controlled at fabrication time by tuning the physical parameters such as inductance, capacitance, Josephson energy, and physical cavity size, depending on the implementation. In conventional quantum computation, transitions between $\ket{\circ}$ and $\ket{\bullet}$ are induced by driving the qubit with active input pulses whose frequencies are calibrated to be resonant with the energy splitting $\Delta$ \cite{krantz2019quantum}. Here, we wish to restrict the need for such pulses to a small set of qubits only at the beginning and/or end of the computation.

At inverse temperature $\beta$, an isolated qubit with Hamiltonian $H_{\rm qubit}$ in thermal equilibrium is in the state $\rho_{\rm qubit} = \frac{1}{1+e^{-\beta \Delta}}(\ket{\circ}\!\bra{\circ}+ e^{-\beta \Delta} \ket{\bullet}\!\bra{\bullet})$; that is, the excited population is roughly $e^{-\beta \Delta}$.  In practice, residual excited-state populations are typically somewhat higher than the equilibrium prediction due to non-equilibrium effects, but remain small~\cite{jin2015thermalExcitation}.

\paragraph{Different excitations in our construction and their energies.}
In our construction, there are several different kinds or flavors of qubits, labeled by $0$, $1$, and $\star$ and by whether they are a carrier, switch, address, or bus register. A design choice we specify here is that all of the particle energies, that is, the energy splitting between the $\ket{\circ}$ and $\ket{\bullet}$ states,  should be equal to the same value $\Delta$, except for  the $0/1$-flavored carrier qubits, which should have energy $2\Delta$. The reason for this choice is that it ensures all of the terms $T$ in the Hamiltonian of \cref{eq:H} preserve total particle energy. Each term either (i) moves particles around without creating or destroying them, or (ii) exchanges two particles of energy $\Delta$ for one particle of energy $2\Delta$.
Formally, we may then write the  ``pinning'' Hamiltonian $P$ as the sum of all the on-site 1-local terms
\begin{equation}
 P={}\underbrace{\Delta\sum_{R\in\{A_1,\ldots,A_n,B\}}\sum_{c\in\bits}n_{R,c}}_{\text{address/bus}}
   +\underbrace{\Delta\sum_{j,v}n_{\car_{j,v},\star}}_{\text{carrier flavor } \star}+\underbrace{2\Delta\sum_{j,v}\sum_{c\in\bits}n_{\car_{j,v},c}}_{\text{carrier flavor } 0,1}
   +\underbrace{\Delta\sum_{v,r}\sum_{c\in\bits}n_{Y_{v,r},c}}_{\text{switch wire}},
 \label{eq:pinning}
\end{equation}
For the following discussion, we assume these pinning terms exist in the Hamiltonian, in particular in the noise resilience analysis in \cref{sec:noise}. We will refer to the physical Hamiltonian as,
\begin{align}
    K = H + P \label{eq:K}
\end{align}
The property we need from $P$ is stated as follows.
\begin{theorem}[Conservation of on-site energy]
\label{thm:commuting}
Every transition $T$ in \cref{sec:hamiltonian_terms} obeys
\begin{equation}
    [P,T]=0,\qquad [P,H]=0.
    \label{eq:pinning-commutes}
\end{equation}
Every valid history has the same pinning energy,
\begin{equation}
    P\ket{\Phi_j(x,b,D)}=(n+2)\Delta\ket{\Phi_j(x,b,D)}.\label{eq:history-pinning-energy}
\end{equation}
Consequently the fixed Hamiltonian $K$ implements the query of
\cref{thm:exact qram} at the same time $\tau$, up to a global phase
$\exp(i\varphi_L-i(n+2)\Delta\tau)$ independent of $x,b,D$, where $\varphi_L$ is the global phase accrued by the Hamiltonian $H$, as defined in \cref{eq:deterministic_state_transfer_eigenvalue_eq}.
\end{theorem}
\begin{proof}
Consider a complete basis of states $\ket{\psi}$ corresponding to all possible configurations of particles in the modes (that is, product states of $\ket{\circ}\!\bra{\circ}$ and $\ket{\bullet}\!\bra{\bullet}$). The operator $P$ is diagonal in this basis since it is 1-local. Furthermore, any term $T$ acts as a permutation on this basis since it is a product of raising and lower operators. To show that $[P,T]=0$, it suffices to show that for any $\ket{\psi}$ and any $T$, $\ket{\psi}$ and $T\ket{\psi}$ have the same eigenvalue of $P$. 

This can be confirmed by inspection of each kind of term. The number of particles of each flavor are preserved some of the terms. For example, the route term of \cref{eq:route} annihilates a carrier particle in one mode and creates a carrier particle of the same flavor in another mode, while doing the same for a particle in the switch wire. The leaf read term of \cref{eq:leaf-read} annihilates a 0/1-flavored particle and creates a 0/1-flavored particle.  This is not true for the exchange terms of \cref{eq:carrier-exchange}, which remove a 0/1-flavored carrier and add a $\star$-flavored carrier and a 0/1-flavored switch wire particle (or vice versa). The fact that the deleted particle has energy $2\Delta$ while the two created particles have energy $\Delta$ ensures the $P$ eigenvalue is conserved. We conclude that $[P,T] = 0$ for all $T$. 

The initial configuration, $\ket{\Phi_0(x,b,D)}$ has excitations in a single carrier register, $n$ address registers, and one bus register. Each of these excitations carry energy $\Delta$. Thus $P\ket{\Phi_0(x,b,D)} = (n+2)\Delta \ket{\Phi_0(x,b,D)}$. All other configurations $\ket{\Phi_0(x,b,D)}$ are formed by applying a sequence of transitions $T$, which conserve the $P$ eigenvalue since $[P,T] = 0$. This proves \eqref{eq:history-pinning-energy}. 
\end{proof}

\paragraph{Comments on how the interactions in our construction could be physically realized.} Our construction requires three kinds of interaction terms: route, exchange, and read. Here we temporarily adopt a more physics-oriented notation $a_u \equiv \sigma^-_u$ and $a_u^\dagger \equiv \sigma^+_u$ for a mode label $u$ (i.e., $u$ denotes the identity as carrier/switch/address/bus and flavor $0$/$1$/$\star$).  The operators $a_u$ and $a_u^\dagger$ are the annihilation and creation operators for an excitation in the mode. Then, the route term has the form
\begin{align}
    a_z^\dagger a_w^\dagger a_v a_u + a_z a_w a_v^\dagger a_u^\dagger
\end{align}
acting on four distinct modes. For our construction, there are two subcases: either (i) all four modes $u,v,w,z$ have on-site energy splitting $\Delta$, or (ii) modes $u$ and $w$ have energy splitting $\Delta$ and modes $v$ and $z$ have energy splitting $2\Delta$. As proved above, in all cases, the term conserves the total on-site energy. The exchange term has the form
\begin{align}
    a_w^\dagger a_u a_v + a_w a_u^\dagger a_v^\dagger
\end{align}
where modes $u$ and $v$ have energy splitting $\Delta$ and $w$ has energy $2\Delta$.

The key observation is that both the route term and the exchange term are \textit{nonlinear}: each is a product of more than two creation and annihilation operators. Quadratic terms of the form $a_w^\dagger a_u + a_w a_u^\dagger$, i.e., beam-splitter interactions, are often easier to engineer, for example with a coupling capacitor or inductor, but by themselves do not make excitations interact. Our construction calls for imprinting the nonlinear couplings directly as static elements of the device. We leave detailed designs for future work, and comment only briefly that these are among the lowest-order nonlinear terms possible. They appear much more tractable to realize than terms involving more than four modes, which a generic $O(1)$-local construction might have required; in Josephson-junction circuits such higher-order terms are additionally suppressed by powers of the zero-point phase fluctuations \cite[Section IIID]{Blais2021circuitQED}. For example, a single Josephson junction shared by four modes produces static four-wave-mixing terms including those required by the route term. However, it also produces cross-Kerr  couplings (e.g., $a_u^\dagger a_u a_v^\dagger a_v$) of comparable strength, which would need to be cancelled with additional circuit elements. Similarly, the exchange term is a standard three-wave-mixing process, which can be realized with a comparably simple circuit of Josephson junctions in a loop threaded by a static magnetic flux \cite{frattini2017threeWaveMixing}. Unlike the four-wave-mixing terms, this requires a symmetry-breaking bias, but the bias can in principle be passive and non-dissipative---for example, flux trapped in a superconducting loop or a global static field---rather than an actively driven control line.  
Importantly, in all cases, a coupling term only has an appreciable effect if the on-site energy it removes equals the energy it creates, i.e., if it commutes with $P$; terms whose energy mismatch is large compared to their strength are effectively suppressed. Requiring $[T,P]=0$ for each desired term $T$ ensures it is energy-conserving, while unwanted terms can be suppressed by choosing $P$ so that they are not.  See Ref.~\cite{Blais2021circuitQED} and references therein for background on these wave-mixing processes and the role of energy matching.
 A complete practical analysis would also need to assess how well fabrication can target the resonance conditions, the desired interaction strengths $J_j$, and the suppression of unwanted terms.

The leaf read term is best understood as a classically controlled linear beam-splitter-type term. That is, depending on the value of the classical bit $\ket{D_x}$, there is either a linear beamsplitter term $a_w^\dagger a_v + a_w a_v^\dagger$ between modes $v$ and $w$  or a similar term $a_z^\dagger a_v + a_z a_v^\dagger$ between modes $v$ and $z$. Our analysis has assumed the classical bit is not subject to noise. For example, each classical bit could be stored as a trapped flux quantum in a superconducting loop, which requires no static power to hold. The resulting static magnetic flux could in turn be used to control the relevant interaction in the quantum device, similar to how on-chip flux storage loops are used to program static biases in quantum annealers~\cite{johnson2010scalable}.
Such bits could be rewritten between queries using single-flux-quantum (SFQ) circuitry \cite{mcdermott2018quantum}.

\paragraph{Why long-range interactions could be plausible.}

Another experimental requirement of our construction is that some interaction terms must couple modes located far apart on the device. This is an engineering challenge rather than an intrinsic barrier. In superconducting or photonic architectures, such a coupling can be mediated by a transmission line or waveguide whose ends are coupled to the two modes \cite{majer2007coupling}. If the two modes share a frequency $\Delta$ that is detuned from every mode of the line by much more than their coupling to it, the line is populated only virtually. The resulting effective interaction between the distant modes is a static beam-splitter-type coupling. For this description to hold, the target interaction strength $J$ must be small compared with the inverse transit time, $J\ell/v \ll 1$, where $\ell$ is the length of the line and $v$ the propagation speed.  This is analogous to the signal-propagation limit on classical RAM and QRAM \cite{Jaques2025qramsurveycritique}, which precludes $O(\log N)$ latency once the device's linear size, growing as $\sqrt{N}$ for a planar layout, becomes large enough. In superconducting and silicon photonic platforms the propagation speed is a substantial fraction of $c$ (roughly $0.25$--$0.4\,c$), so this limit would not necessarily be the bottleneck at realistic device sizes. Our construction as written requires long-range 4-local interactions, but if needed it could be modified so that all of the 4-local interaction terms are also geometrically local, whereas long-range interactions are contained entirely to 2-local beamsplitter interactions.

\section{Noise resilience}
\label{sec:noise}

In practice, one cannot expect a perfect replica of the Hamiltonian evolution written down in \cref{thm:exact qram}. This is due to three main reasons: perfect initialization is expensive, there will be noise from the environment during evolution, and the Hamiltonian terms cannot be implemented perfectly, there will be some imprecision. We give resilience guarantees for all three sources of noise. All guarantees rest on a single structural fact: for a fixed computational basis query, the qubits that can affect the logical output are confined to a small
\emph{footprint}, and every part of $H$ that could carry a stray carrier from outside the
footprint onto the queried branch is annihilated on the clean sector. 

\paragraph{Assumptions.}
The memory register $M$ is a noiseless classical memory. Corruption of stored data is a separate concern common to every QRAM and is
out of scope. The address register $A$ and bus $B$ are \emph{not} auxiliary qubits. Noise acts only on the
auxiliary carrier and switch qubits, which are exactly the qubits the thermal model
populates. The thermal and continuous-noise models act
only on the auxiliary qubits.  The static-imprecision model is more general and also
allows parasitic terms that touch the address, bus, or memory.

Throughout this section, $A,B$ denote the physical address and bus
ports, and $A_{\rm log},B_{\rm log}$ denote the algorithm's address
and bus registers. Encoding transfers their values into the physical
ports, and decoding returns the output to the algorithm.

\paragraph{Encoding and decoding.}
The algorithm uses ordinary qubits, while each physical address or
bus register uses two qubits. The encoding isometry for one register is
\[
    V_1
      =\ket{\bullet\circ}\!\bra{0}
       +\ket{\circ\bullet}\!\bra{1}.
\]
Let $V_B=V_1$ encode the bus and let
$V_{AB}=V_1^{\otimes(n+1)}$ encode the address and bus.
The corresponding encoding channel is
$\mathcal V_{AB}(\rho)=V_{AB}\rho V_{AB}^\dagger$.

To define decoding even when noise takes the ports outside their
encoding subspace, let $\Pi_{\rm enc}=V_{AB}V_{AB}^\dagger$ and set
\[
    \mathcal D_{AB}(\rho)
      =V_{AB}^\dagger\rho V_{AB}
       +\operatorname{tr}\!\left[(I-\Pi_{\rm enc})\rho\right]
        \ket{0^{n+1}}\!\bra{0^{n+1}}.
\]
This channel reverses the encoding on valid inputs. It maps the
component outside the encoding subspace to a fixed output and
discards coherence between the two subspaces.
Unless stated otherwise, its output registers are
$A_{\rm log},B_{\rm log}$. All these maps act as the identity on
registers not displayed.

\paragraph{Recorded address interface.}
An Exchange term temporarily moves an address bit out of its physical input
register.  For the general noise bounds, it is useful to retain a copy of the
address at the root.

For this interface, the algorithm's address register $A_{\rm log}$
is retained at the root. We denote this retained register by $R$.
The ideal Hamiltonian does not act on $R$. Before the query,
coherently copy its computational-basis value into the physical
address ports:
\begin{equation}
    W_{\rm rec}\ket{x}_R\ket{b}_{B_{\rm log}}
      =
    \ket{x}_R\ket{x}_A V_B\ket{b}_{B_{\rm log}} .
    \label{eq:recorded-input}
\end{equation}
This is not a measurement.  For a superposition of addresses, $R$ and the
physical address ports are entangled.
After the query, decode the bus into $B_{\rm log}$ and the physical
address into a temporary register $\bar A$. The algorithm's address
remains in $R$. Since the ideal query returns the same address in
$\bar A$, a CNOT from $R$ clears this temporary register:
\begin{equation}
    \mathcal O_{\rm rec}(\rho)
      =
    \operatorname{tr}_{\bar A}\!\left[
      \mathrm{CNOT}_{R\to\bar A}\,
      (\operatorname{id}_R\otimes\mathcal D_{AB})(\rho)\,
      \mathrm{CNOT}_{R\to\bar A}^{\dagger}
    \right].
    \label{eq:recorded-output}
\end{equation}
For an ideal query, $\bar A$ contains the same address as $R$,
so the CNOT leaves $\bar A$ in $\ket{0^n}$ and uncorrelated with
the other registers.
The resulting logical operation is
\begin{equation}
    U_D
      =
    \sum_x \proj{x}_R\otimes X_{B_{\rm log}}^{D_x}.
\end{equation}

The Hamiltonian \(K\) does not act on the retained address register \(R\). We use this
interface only for the static Hamiltonian term imprecision and continuous-noise bounds below.

\paragraph{Physical Hamiltonian.}
Throughout this section we keep the pinning Hamiltonian on.  Let
\begin{equation}
    K=H+P,
\end{equation}
where $P$ is defined in \cref{eq:pinning}.
By \cref{thm:commuting},
\begin{equation}
    [H,P]=0,
\end{equation}
so $K$ implements the same ideal QRAM evolution as $H$, up to a global
phase on every valid history.

Fix a classical database $D$.  Since the memory is unchanged by the ideal
query, its computational-basis sectors are invariant.  We write
\begin{equation}
    K_D = H_D + P
    \label{eq:physical-hamiltonian-D}
\end{equation}
for the restriction of the physical Hamiltonian to the memory sector
$\ket D_M$.  All statements below are made in a fixed $D$ sector unless
stated otherwise.

\paragraph{The plan.}
We first show that a query can be isolated from qubits outside its
branch. We then use this fact to bound errors from thermal
initialization, imperfect Hamiltonian terms, and noise during evolution.

\subsection{Branch isolation}
\label{sec:branch-isolation}

All three resilience bounds use the same observation: although the device has
$\Theta(N)$ auxiliary qubits, a computational-basis query only visits
$\Theta(n^2)$ of them.  More strongly, for a pair of addresses $x,y$,
everything outside the union of their visited qubits evolves as a common
tensor factor.

\paragraph{Branch footprints.}
Fix $x\in\bits^n$.  Its \emph{footprint} $\Foot_x$ contains all three qubits of
every carrier register visited by the ideal history for $x$, together with both
qubits of every switch slice on the root-to-$x$ path.  The physical address and
bus ports are not included in $\Foot_x$; we keep them explicitly below.  Since
the schedule has length $L=2n^2+4n+3$ and a switch at depth $d$ has
$4(n-d)-1$ slices,
\begin{equation}
    F_n:=|\Foot_x|
      =3(L+1)+2\sum_{d=0}^{n-1}\bigl(4(n-d)-1\bigr)
      =10n^2+14n+12.
    \label{eq:footprint-size}
\end{equation}
Thus $F_n=\Theta(n^2)=\Theta(\log^2N)$, even though the full device has
$\Theta(N)$ auxiliary qubits.

For an address $z$ and bus value $b$, let
$\ket{\widehat\Phi_j(z,b)}$ denote the ideal history state restricted to the
physical root ports $A,B$ and the auxiliaries in $\Foot_z$, with the fixed
memory state and the all-empty outside auxiliaries omitted. The root ports are
retained because exchange terms change their state during the query.  Write
$\ket\circ_S$ for the all-empty unary state on a set $S$ of auxiliary qubits.

\begin{lemma}[Two-address separation]
\label{lem:two-address-separation}
Fix $x,y\in\bits^n$, set $\Foot_{xy}=\Foot_x\cup\Foot_y$, and let
$z\in\{x,y\}$.  The subspace
\begin{equation}
\begin{split}
    \mathcal S_z^{xy}
      ={}&\operatorname{span}\{\ket{\widehat\Phi_j(z,b)}:
          0\le j\le L,\ b\in\bits\}\\
       &\otimes\ket\circ_{\Foot_{xy}\setminus\Foot_z}
       \otimes\mathcal H_{{\rm aux}\setminus\Foot_{xy}}
\end{split}
\label{eq:pair-sector}
\end{equation}
reduces both $H_D$ and $K_D$.  On this subspace,
\begin{align}
    H_D
      &=H_{{\rm br},z}\otimes I+I\otimes H_{{\rm out},D}^{xy},
      \label{eq:two-address-split}\\
    K_D
      &=\bigl(H_{{\rm br},z}+(n+2)\Delta I\bigr)\otimes I
        +I\otimes K_{{\rm out},D}^{xy},
      \label{eq:two-address-pinned-split}
\end{align}
where $H_{{\rm br},z}$ is the weighted history-path Hamiltonian for address
$z$, and the outside operators act only on auxiliaries outside $\Foot_{xy}$.
The outside operators are the same for $z=x$ and $z=y$, with
\begin{equation}
    K_{{\rm out},D}^{xy}
      =H_{{\rm out},D}^{xy}
       +\sum_{q\in{\rm aux}\setminus\Foot_{xy}}\Delta_q n_q.
\end{equation}
Any retained root address record $R$ is an additional untouched tensor factor.
\end{lemma}

\begin{proof}
Inspect the transition terms from \cref{sec:hamiltonian_terms}.  On an intended
history state for $z$, the unique successor and predecessor terms generate
$H_{{\rm br},z}$.  Every other term supported on the visited part of the
schedule contains a lowering operator on an empty carrier or switch qubit, or
on a carrier/input rail of the wrong flavor, and therefore vanishes.

A term whose auxiliary support is disjoint from $\Foot_{xy}$ acts only on the
outside factor.  The sum of these terms is independent of whether the inside
branch is $x$ or $y$.

It remains to rule out couplings between the two factors.  Root
absorption/return, store/unstore, and leaf-read terms are co-located with their
carrier registers, so at an unused location a carrier lowering operator kills
the state.  Only a route term can cross the footprint boundary.  If the stored
direction at a visited node $v$ is $c=z_{|v|+1}$, then the route toward the
other child lowers the empty switch qubit $q^{1-c}_{v,r}$; its adjoint lowers
the equally empty $q^{1-c}_{v,r+1}$.  Hence both directions vanish.  The same
argument applies after $x$ and $y$ diverge.

This proves the tensor decomposition for $H_D$.  By
\cref{thm:commuting}, the pinning Hamiltonian has the constant value
$(n+2)\Delta$ on every valid inside history, while all pinning terms outside the
union remain in the common outside factor.  This gives
\eqref{eq:two-address-pinned-split}.  Hermiticity gives reduction.
\end{proof}

\paragraph{Exact evolution of one address pair.}
Write $U_x=X^{D_x}$ for the operation on the bus conditioned on address
$x$, so $U_D=\sum_x\proj{x}_A\otimes U_x$.  Let
$\ket{\mathrm{clean}}_{\Foot_{xy}}$ be the restriction of the ideal auxiliary
input to $\Foot_{xy}$: the initial root carrier is $\ket\star$ and every other
carrier and switch qubit is empty.  The separation lemma and exact transfer
then give, for every operator $Z$ on the bus and an untouched reference
system, and every operator $\omega$ on the outside qubits,
\begin{equation}
\begin{split}
&\mathcal D_{AB}\operatorname{tr}_{\rm aux}\!\left[
 e^{-iK_D\tau}\Bigl(
   \mathcal V_{AB}(\ket x\!\bra y_A\otimes Z)
      \otimes\proj{\star}_{\car_{0,\mathcal R}}
\otimes\proj{\circ}_{\Foot_{xy}\setminus\car_{0,\mathcal R}}
   \otimes\omega\Bigr)e^{iK_D\tau}\right]\\
&\hspace{32mm}
 =\ket x\!\bra y_A\otimes U_xZU_y^\dagger\,\operatorname{tr}(\omega).
\end{split}
\label{eq:pairwise-exact}
\end{equation}
The perfect-transfer phase and the pinning phase are branch independent,
so both cancel between ket and bra.  Notice that
\eqref{eq:pairwise-exact} concerns the complete encoded query; the physical
address ports need not preserve the address at intermediate times.

\paragraph{Perturbation estimate.}
We will repeatedly use Duhamel's formula.  If $\mathcal L$ and $\mathcal L'$ generate
finite-dimensional channel semigroups, then
\begin{equation}
    e^{t\mathcal L}-e^{t\mathcal L'}
      =\int_0^t e^{(t-s)\mathcal L}
       (\mathcal L-\mathcal L')e^{s\mathcal L'}\,ds,
    \label{eq:duhamel-channel}
\end{equation}
and therefore
\begin{equation}
    \bigl\|e^{t\mathcal L}-e^{t\mathcal L'}\bigr\|_\diamond
      \le t\|\mathcal L-\mathcal L'\|_\diamond.
    \label{eq:duhamel-bound}
\end{equation}
For Hamiltonians $K,K'$, the generator difference is
$-i[K-K',\,\cdot\,]$ and has diamond norm at most $2\|K-K'\|$.

\subsection{Thermal initialization and reuse}
\label{sec:thermal}

Define the Gibbs state and auxiliary marginal of the physical Hamiltonian $K_D$ by
\begin{equation}
    \Gamma_{\beta,D}=\frac{e^{-\beta K_D}}
       {\operatorname{tr}(e^{-\beta K_D})},\qquad
    \rho_{\beta,D}=\operatorname{tr}_{AB}\Gamma_{\beta,D}.
    \label{eq:interacting-gibbs}
\end{equation}
The trace is over all $2(n+1)$ physical root-port qubits, including states outside the input code space.
Thus $\rho_{\beta,D}$ is the idle state of the
carrier and switch qubits before a query.  We never turn off $H$ or $P$.
We assume that the device can reach this Gibbs state.  We do not bound the
thermalization time.

Our goal is to prepare this state once and then reuse the device.  Let
$\mathcal B_{\rm in}$ load the address and bus into the root ports and prepare
the input carrier in $\ket\star$.  Let $\mathcal B_{\rm out}$ extract the
returned address and bus and reset the input and output carrier registers.
Both maps act only at the root.  One query is the channel
\begin{equation}
    \mathcal C_D(\rho)
      =\mathcal B_{\rm out}\!\left(
        e^{-iK_D\tau}\,
        \mathcal B_{\rm in}(\rho)\,
        e^{iK_D\tau}
      \right).
    \label{eq:full-query-cycle}
\end{equation}
The output of $\mathcal C_D$ retains the bulk auxiliary state.

Let
\begin{equation}
    \kappa
      =\max_{q\in{\rm aux}}
        \sum_{\alpha:q\in\supp(h_\alpha)}J_\alpha
      \le4J_*,
    \qquad J_*=\max_jJ_j,
    \qquad g=\Delta-\kappa,
    \label{eq:incident-coupling-bound}
\end{equation}
where $\Delta$ is the smaller on-site gap.

\begin{theorem}[Thermal initialization]
\label{thm:thermal}
Assume $g>0$ and fix the database $D$.  Consider any algorithm that makes
$m$ QRAM queries, with arbitrary channels on its registers between queries
and no operation on the bulk auxiliaries.  Prepare the auxiliaries once in
$\rho_{\beta,D}$, and use \eqref{eq:full-query-cycle} for every query.  Let
$\mathcal A^{\rm phys}_{m,D}$ be the resulting induced channel and
$\mathcal A^{\rm ideal}_{m,D}$ the same algorithm with ideal QRAM queries.
Then
\begin{equation}
    \left\|
      \mathcal A^{\rm phys}_{m,D}
      -\bigl[\eta\mapsto
       \mathcal A^{\rm ideal}_{m,D}(\eta)\otimes\rho_{\beta,D}\bigr]
    \right\|_\diamond
      \le
      \min\left\{2,4m\sqrt{F_ne^{-\beta g}}\right\}.
    \label{eq:multi-query-bound}
\end{equation}
Hence error at most $\varepsilon$ follows from
\begin{equation}
    \beta g\ge\log F_n+2\log(4m/\varepsilon).
    \label{eq:multi-query-temperature}
\end{equation}
For $m,\varepsilon^{-1}=\poly(n)$ and constant $g$, this requires
$T=O(1/\log n)=O(1/\log\log N)$.
\end{theorem}

\begin{proof}
Set $\delta=F_ne^{-\beta g}$.  It suffices to prove that one query changes
the joint query--device state by at most $4\sqrt\delta$, then sum this error
over the $m$ queries.

Fix an address $x$, and let $\Pi_x$ project its footprint $\Foot_x$ onto the
vacuum.  As in \cref{lem:two-address-separation}, $\Pi_x$ reduces $K_D$.
Let $N_x=\sum_{q\in\Foot_x}n_q$.
For each Hamiltonian term $T$,
\[
    T+T^\dagger\succeq-(T^\dagger T+TT^\dagger).
\]
Summing these bounds and using the definition of $\kappa$ gives
\begin{equation}
    K_D\succeq K_{{\rm out},D}^x+P_{AB}+gN_x,
    \label{eq:thermal-energy-comparison}
\end{equation}
where $P_{AB}$ is the pinning Hamiltonian on the root ports.  On the
$\Pi_x$ sector, equality holds with $K_{{\rm out},D}^x+P_{AB}$.  The
min--max principle then gives
\begin{equation}
    1-\operatorname{tr}(\Pi_x\rho_{\beta,D})
      \le F_ne^{-\beta g}
      =\delta.
    \label{eq:thermal-delta}
\end{equation}

Before loading the root ports, define
$\Pi=\sum_x\proj{x}_{A_{\rm log}}\otimes\Pi_x$.
For an arbitrary input $\eta$, including
a reference, put
\[
    \sigma=\eta\otimes\rho_{\beta,D},
    \qquad
    \sigma_{\rm good}=\Pi\sigma\Pi.
\]
Then $\operatorname{tr}(\Pi\sigma)\ge1-\delta$.  Since every $\Pi_x$
commutes with $\rho_{\beta,D}$,
\begin{equation}
    \Pi_x\rho_{\beta,D}\Pi_y
      =\proj\circ_{\Foot_x\cup\Foot_y}\otimes\omega_{xy},
    \qquad
    [\omega_{xy},K_{{\rm out},D}^{xy}]=0.
    \label{eq:thermal-stationary-outside}
\end{equation}
The first identity holds because the vacuum union reduces $K_D$.  The second
holds because the outside block is proportional to
$e^{-\beta K_{{\rm out},D}^{xy}}$.

By \cref{lem:two-address-separation,thm:commuting}, each good $(x,y)$ block
therefore undergoes the exact query while its outside state stays fixed:
\begin{equation}
    \mathcal C_D(\sigma_{\rm good})
      =(\mathcal U_D\otimes\operatorname{id}_{\rm aux})
        (\sigma_{\rm good}).
    \label{eq:thermal-good-exact}
\end{equation}
The gentle measurement lemma~\cite{Winter1999} gives
$\|\sigma-\sigma_{\rm good}\|_1\le2\sqrt\delta$.  Contractivity then gives
\begin{equation}
    \left\|
      \mathcal C_D(\sigma)
      -(\mathcal U_D\otimes\operatorname{id}_{\rm aux})(\sigma)
    \right\|_1
      \le4\sqrt\delta.
    \label{eq:one-query-joint-bound}
\end{equation}

Now use $m$ adjacent hybrids.  In each hybrid, the ideal prefix leaves the
auxiliaries exactly in $\rho_{\beta,D}$ and uncorrelated with the algorithm.
Replacing the next ideal query by a physical query costs at most
$4\sqrt\delta$.  Later channels cannot increase the trace norm.  Summing over
the $m$ queries proves \eqref{eq:multi-query-bound}.  Rearranging gives
\eqref{eq:multi-query-temperature}.
\end{proof}

\paragraph{Energy per query.}
The theorem shows that we need no bulk reset between these queries.  To count
energy, realize input and output transfer by swaps between the physical root
ports and external registers with the same two-qubit encoding.
    The device Hilbert space then stays
fixed, and the boundary operations touch $O(n)$ root qubits in total.

\begin{corollary}[Quantum boundary energy per query]
\label{cor:boundary-energy}
The sum of the absolute changes in expected device energy caused by the root
operations in one query is
\begin{equation}
    O\bigl(n(\Delta+J_*)\bigr).
    \label{eq:boundary-exchange}
\end{equation}
Free evolution under $K_D$ conserves expected device energy.
If each root qubit operation costs some fixed constant energy, this gives $O(\log N)$ energy per query.
\end{corollary}

\begin{proof}
Consider one root operation supported on device qubits $S$.  Write
$K_D=K_{\rm away}+K_{\rm touch}$, where $K_{\rm away}$ contains the terms
whose support misses $S$.  The operation leaves the expectation of
$K_{\rm away}$ unchanged.  Hence it changes the expected device energy by at
most
\[
    2\|K_{\rm touch}\|
      =O\bigl(|S|(\Delta+J_*)\bigr).
\]
Each root qubit has on-site energy at most $2\Delta$ and belongs to at most
four couplings of norm $O(J_*)$.  Summing over the root operations proves
\eqref{eq:boundary-exchange}.  Free evolution conserves energy because its
generator is $K_D$.  The interface uses $O(n)$ bounded-size root operations,
which gives the energy bound.
\end{proof}

\paragraph{Scope.}
Classical database storage, access, and updates are outside this accounting.
The one-time thermal preparation may have extensive cost.  The theorem assumes
ideal coherent evolution at fixed $D$.  The next two subsections give worst-case error guarantees arising from error in the evolution of $K$ during $m$ sequential queries without cooling,  in the form of a bound on the diamond norm of the combined channel.

\subsection{Imperfect Hamiltonian terms}
\label{sec:imperfect}

We now allow static errors in the manufactured Hamiltonian.  The errors may
have arbitrary support, so the model includes crosstalk.  The crucial quantity will be the error weight that meets one
query branch or footprint. We will analyse the error after $m$ queries. 
An error can damage a branch before the algorithm visits it, so query $k$ pays for the first $k$
query intervals.  This gives the quadratic dependence on $m$.

We prepare vacant auxiliaries and memory $\ket D$ once and use the recorded
interface.  Root controls are ideal.  Between queries, the algorithm may act
arbitrarily on its own registers, but not on the retained device.  The target
remains the original database $D$.  Let
$\epsilon_m=\|\mathcal A^{\rm phys}_{m,D}-\mathcal A^{\rm ideal}_{m,D}\|_\diamond$.

Suppose
\begin{equation}
    \widetilde K=K+\sum_\alpha E_\alpha,
    \qquad E_\alpha=E_\alpha^\dagger,
    \qquad \|E_\alpha\|\le\delta_\alpha,
    \label{eq:static-error-model}
\end{equation}
with arbitrary supports $S_\alpha=\supp(E_\alpha)$.  Let $\mathsf P_{AB}$
be the physical address and bus ports, and define
\begin{equation}
\begin{split}
    \mathsf F_x^+
      &=\Foot_x\cup\mathsf P_{AB}\cup\{R_1,\ldots,R_n,M_x\},\\
    W&=\max_x\sum_{\alpha:S_\alpha\cap\mathsf F_x^+\ne\emptyset}\delta_\alpha.
\end{split}
    \label{eq:branch-term-weight}
\end{equation}
Thus $|\mathsf F_x^+|=F_n+3n+3$.

\begin{corollary}[Repeated queries with static imprecision]
\label{cor:terms}
For every $m\ge1$,
\begin{equation}
    \epsilon_m\le\min\{2,\,2m(m+1)\tau W\}.
    \label{eq:term-error-bound}
\end{equation}
If $\lambda_{\rm err}=\max_q\sum_{\alpha:q\in S_\alpha}\delta_\alpha$, then
\[
    \epsilon_m=O(m^2\lambda_{\rm err}\log^4N).
\]
Hence, if each qubit meets at most $d_{\rm err}$ error terms of norm at most
$\delta$, then
$\delta=O(\varepsilon/(m^2d_{\rm err}\log^4N))$ suffices for error at most
$\varepsilon$.
\end{corollary}

\begin{proof}
We compare each physical query with an ideal query followed by one noisy idle
interval.  Split $W=a+w$, where $a$ is the weight of errors touching the
record $R$ or physical ports $\mathsf P_{AB}$, and $w$ is the largest
remaining branch weight.  Removing the first class for all $m$ intervals
costs at most $2m\tau a$ by Duhamel.

For the comparison, keep the ports vacant and evolve the retained device with
\begin{equation}
    K_{\rm idle}=\bra{\circ}_{AB}K\ket{\circ}_{AB}
      +\sum_{\alpha:\,S_\alpha\cap(R\cup\mathsf P_{AB})=\emptyset}E_\alpha.
    \label{eq:idle-comparison-hamiltonian}
\end{equation}
Apply the same root resets, but inject no carrier.  Let $\ket{\xi_k}$ purify
this idle state after $k$ intervals.

Consider query $k$.  In record block $x$, delete every remaining error that
meets $\Foot_x\cup\{M_x\}$ during the first $k$ intervals.  Duhamel changes
either the active or idle evolution by at most $k\tau w$.  After deletion the
footprint stays vacant and $M_x$ still stores $D_x$.  Branch separation makes
the active query exact, with the same outside evolution in both processes.
Thus, for a dilation $V_k$ of the active query,
\begin{equation}
    \bigl\|V_k(I\otimes\ket{\xi_{k-1}})
       -U_D\otimes\ket{\xi_k}\bigr\|\le2k\tau w.
    \label{eq:static-idle-amplitude}
\end{equation}
The retained address register $R$ lets us choose which errors to
remove separately for each address, without measuring it.
The resulting operator is block diagonal in $R$, so its norm is
the largest norm among the address blocks, at most $w$.

Replace the $m$ queries one at a time.  Equation
\eqref{eq:static-idle-amplitude} changes the final induced channel by at most
$4k\tau w$ at query $k$.  Hence
\[
    \epsilon_m
      \le2m\tau a+4\tau w\sum_{k=1}^m k
      \le2m(m+1)\tau W.
\]
Finally,
$W\le(F_n+3n+3)\lambda_{\rm err}$ by summing incident error weights over
$\mathsf F_x^+$.  Since $F_n,\tau=\Theta(n^2)$, the claimed scaling follows.
\end{proof}

\subsection{Local Markovian noise during the query}
\label{sec:dynamical}

We now allow Markovian noise throughout the $m$ queries.  Our technique will be similar to \cref{sec:imperfect}.
The only new point is that
an error may leave the idle state and return later, so we track this return
amplitude.  This gives a cubic dependence on $m$.

Let
\begin{equation}
    \mathcal L_t(\rho)=-i[K,\rho]+\sum_\alpha\mathcal N_{\alpha,t}(\rho),
    \label{eq:lindblad}
\end{equation}
where
\begin{equation}
\begin{split}
    \mathcal N_{\alpha,t}(\rho)
      ={}&-i[K_{\alpha,t},\rho]\\
       &+\sum_\ell\left(L_{\alpha\ell,t}\rho L_{\alpha\ell,t}^\dagger
          -\tfrac12\{L_{\alpha\ell,t}^\dagger L_{\alpha\ell,t},\rho\}\right).
\end{split}
    \label{eq:local-markovian-component}
\end{equation}
All operators in component $\alpha$ have support $S_\alpha$.  Put
\begin{equation}
    \gamma_\alpha(t)=\|K_{\alpha,t}\|
       +\sum_\ell\|L_{\alpha\ell,t}\|^2,
    \qquad
    \|\mathcal N_{\alpha,t}\|_\diamond\le2\gamma_\alpha(t).
    \label{eq:markovian-component-rate}
\end{equation}
For $I_j=[(j-1)\tau,j\tau]$, define
\begin{equation}
    G_j=\max_x\int_{I_j}
       \sum_{\alpha:S_\alpha\cap\mathsf F_x^+\ne\emptyset}
          \gamma_\alpha(t)\,dt,
    \qquad G=\max_jG_j.
    \label{eq:general-branch-exposure}
\end{equation}

\begin{theorem}[Repeated queries with Markovian noise]
\label{thm:dynamical}
Under the assumptions of \cref{sec:imperfect},
\begin{equation}
    \epsilon_m\le
    \min\left\{2,\,8m\sum_{j=1}^m(m-j+1)G_j\right\}
    \le\min\{2,\,4m^2(m+1)G\}.
    \label{eq:dynamical-bound}
\end{equation}
If
$\lambda_{\rm dyn}(t)=
\max_q\sum_{\alpha:q\in S_\alpha}\gamma_\alpha(t)$, then
\begin{equation}
    G_j\le(F_n+3n+3)\int_{I_j}\lambda_{\rm dyn}(t)\,dt.
    \label{eq:dynamical-degree-exposure}
\end{equation}
Hence $\lambda_{\rm dyn}(t)\le\gamma$ gives
$\epsilon_m=O(m^3\gamma\log^4N)$.  If each qubit meets at most
$d_{\rm dyn}$ components of rate at most $\gamma_0$, then
$\gamma_0=O(\varepsilon/(m^3d_{\rm dyn}\log^4N))$ suffices for error at
most $\varepsilon$.
\end{theorem}

We use one elementary composition fact.

\begin{fact}[Return-amplitude composition]
\label{fact:idle-amplitude-composition}
Let $V_k$ be an isometry describing query $k$, including its
environment, and let $\ket{\xi_k}$ be normalized idle states
of the device and environment. Define
\[
    Z_k
      =(I\otimes\bra{\xi_k})\,V_k\,
        (I\otimes\ket{\xi_{k-1}}).
\]
Thus $Z_k$ is the operator on the algorithm's registers obtained
by starting the device and environment in $\ket{\xi_{k-1}}$
and taking their overlap with $\ket{\xi_k}$ after the query.
If $\|Z_k-U_k\|\le r_k$ for the ideal query unitaries $U_k$, then the final
error is at most
\begin{equation}
    \min\left\{2,\,4\left(\sum_{k=1}^m\sqrt{r_k}\right)^2\right\}.
    \label{eq:return-amplitude-composition}
\end{equation}
This remains true with arbitrary channels between queries.
\end{fact}

\begin{proof}
Write the joint state after query $k$ as
\[
    \ket{\Psi_k}
      =\ket{u_k}\otimes\ket{\xi_k}+\ket{v_k},
    \qquad
    (I\otimes\bra{\xi_k})\ket{v_k}=0.
\]
Here $\ket{u_k}$ is the possibly unnormalized component along
the idle device--environment state, and $\ket{v_k}$ is
the orthogonal remainder. Let
$\ell_k=\|\ket{v_k}\|$ and
$d_k=\|\ket{u_k}-\ket{u_k^{\rm id}}\|$, where
$\ket{u_k^{\rm id}}$ is the ideal algorithm state.
Since $V_k$ is an isometry,
\[
    \ell_k\le\ell_{k-1}+\sqrt{2r_k},
    \qquad
    d_k\le d_{k-1}+r_k+\sqrt{2r_k}\,\ell_{k-1}.
\]
Hence $d_m\le(\sum_k\sqrt{r_k})^2$ and
$\ell_m\le\sqrt{2}\sum_k\sqrt{r_k}$.
Let $\rho_m^{\rm phys}$ and $\rho_m^{\rm ideal}$ be the final
algorithm states, including any reference.
Tracing out the device and environment removes the cross terms
between $\ket{u_m}\otimes\ket{\xi_m}$ and $\ket{v_m}$ because
$(I\otimes\bra{\xi_m})\ket{v_m}=0$. Therefore,
\[
    \|\rho_m^{\rm phys}-\rho_m^{\rm ideal}\|_1
      \le 2d_m+\ell_m^2
      \le 4\left(\sum_{k=1}^m\sqrt{r_k}\right)^2.
\]
Here the first term bounds the change in the rank-one operator
$\proj{u_m}$, and the second is the trace of the remaining
positive operator. Maximizing over the input and reference
gives the diamond-norm bound, capped at $2$.
\end{proof}

\begin{proof}[Proof of \cref{thm:dynamical}]
Remove all noise touching the record $R$ or physical ports $\mathsf P_{AB}$.
Let $a_j$ be its exposure in $I_j$, let $b_j=G_j-a_j$, and set
$B_k=\sum_{j=1}^kb_j$.  Removing these components costs at most
$2\sum_j a_j$.

Run the idle comparison from \eqref{eq:idle-comparison-hamiltonian} with the
remaining noise. Let $\ket{\xi_k}$ be the joint device--environment state after
$k$ idle intervals, which we call the idle state.
Let $V_k$ be an isometry describing the actual query $k$,
including its interaction with the environment.  We claim
\begin{equation}
    \left\|
      (I\otimes\bra{\xi_k})\,V_k\,
      (I\otimes\ket{\xi_{k-1}})
      -U_D
    \right\|
      \le 2B_k.
    \label{eq:markovian-idle-amplitude}
\end{equation}
Fix record value $x$.  Coherently compare the active query with an idle
interval, and delete all noise meeting $\Foot_x\cup\{M_x\}$ during the first
$k$ intervals.  Its exposure is at most $B_k$.  The footprint then stays
clean, so branch separation makes the active arm exact, while both arms have
the same outside evolution.  Duhamel changes their off-diagonal block by at
most $2B_k$.  Since the record blocks are orthogonal, maximizing over $x$
proves \eqref{eq:markovian-idle-amplitude} without summing over addresses.

Apply \cref{fact:idle-amplitude-composition} with $r_k=2B_k$ and restore the
removed root noise:
\begin{align}
    \epsilon_m
      &\le2\sum_j a_j+8\left(\sum_k\sqrt{B_k}\right)^2\notag\\
      &\le2\sum_j a_j
         +8m\sum_{j=1}^m(m-j+1)b_j\notag\\
      &\le8m\sum_{j=1}^m(m-j+1)G_j.
    \label{eq:dynamical-compose}
\end{align}
The middle step is Cauchy--Schwarz.  Bounding $G_j$ by $G$ proves
\eqref{eq:dynamical-bound}.  Summing incident rates over $\mathsf F_x^+$ gives
\eqref{eq:dynamical-degree-exposure}.  Since $F_n,\tau=\Theta(n^2)$, the
stated scaling follows.
\end{proof}

\bibliographystyle{alphaurl}
\bibliography{refs}          %

\end{document}